\documentclass[11pt,reqno]{amsart}
\usepackage{amsmath,amssymb,hyperref,mathrsfs,graphicx}
\usepackage[utf8x]{inputenc}
\usepackage[T1]{fontenc}
\usepackage{amsmath}
\usepackage{amsfonts}
\usepackage{latexsym}
\usepackage{cite}
\usepackage{float}

\usepackage{mathtools}      
\mathtoolsset{showonlyrefs} 

\usepackage{amsthm}
\usepackage{amssymb,amscd}
\usepackage{mathtools}
\usepackage{xargs}
\usepackage{graphicx}
\usepackage{color}
\usepackage{enumitem}
\usepackage{tikz}
\usepackage{caption,subcaption} 

\usepackage{lmodern}

\usepackage[text={6in, 8in}, centering]{geometry}
 
\usepackage{parskip}
\usepackage{microtype}

\usepackage{hyperref}

\newtheorem{thm}{Theorem}[section]
\newtheorem{conj}[thm]{Conjecture}
\newtheorem{prop}[thm]{Proposition}
\newtheorem{lem}[thm]{Lemma}

\newtheorem{cor}[thm]{Corollary}
\newtheorem{rmk}[thm]{Remark}

\newtheorem{thmA}{Theorem}
\newtheorem{corA}[thmA]{Corollary}
\newtheorem{propA}[thmA]{Proposition}

\theoremstyle{definition}

\newtheorem{defn}[thm]{Definition}

\newcommand{\T}{\mathbb{T}}
\newcommand{\R}{\mathbb{R}}

\newcommand*\circled[1]{\tikz[baseline=(char.base)]{
            \node[shape=circle,draw,inner sep=2pt] (char) {#1};}}

\title[]{An analytical study of flatness and intermittency through Riemann's non-differentiable functions}

\author[D. Eceizabarrena]{Daniel Eceizabarrena}
 \address[D. Eceizabarrena]{Department of Mathematics and Statistics, Univeristy of Massachusetts Amherst, Amherst, MA 01003-9305, United States.}
\email{eceizabarrena@math.umass.edu}

\author[V. Vila\c ca Da Rocha]{Victor Vila\c ca Da Rocha}
 \address[V. Vila\c ca Da Rocha]{School of Mathematics, Georgia Institute of Technology, Atlanta, GA 30332-0160, United States.}
\email{vrocha3@gatech.edu}

\date{\today}

\begin{document}

\begin{abstract}
In the study of turbulence, intermittency is a measure of how much Kolmogorov's theory of 1941 deviates from experimental measurements. It is quantified with the flatness of the velocity of the fluid, usually based on structure functions in the physical space. However, it can also be defined with Fourier high-pass filters. Experimental and numerical simulations suggest that the two approaches do not always give the same results. Our purpose is to compare them from the analytical point of view of functions. We do that by studying generalizations of Riemann's non-differentiable function, yielding computations that are related to some classical problems in Fourier analysis. The conclusion is that the result strongly depends on regularity. To visualize this, we establish an analogy between these generalizations and the influence of viscosity in turbulent flows. 
This article is motivated by the mathematical works on the multifractal formalism and the discovery of Riemann's non-differentiable function as a trajectory of polygonal vortex filaments.
\end{abstract}


\maketitle

\section{Introduction}

Small-scale intermittency is one of the properties observed in turbulence and it measures how much experimental results deviate from Kolmogorov's statistical theory of 1941 (shortened K41). 
Given that turbulence is characterized by the irregular motion of the fluid, this theory is based in studying the variations of the velocity field in small scales. 
In the three dimensional setting, let us denote the position by $x \in \mathbb R^3$ and the velocity vector by $u(x) \in \mathbb R^3$.
Then, the velocity increments are
\begin{equation}\label{Velocity_Increments}
    \delta_l u(x) = u(x+l) - u(x), \qquad \text{ for small } l .
\end{equation}
If the flow is supposed to be homogeneous, the statistical properties of these increments do not depend on the position~$x$. Besides, if the flow is isotropic, the increments depend only on the scale parameter $\ell = |l|$, not on the direction of $l$. 
This way, the structure functions 
\begin{equation}\label{StructureFunctionPhysical}
    S_p(\ell) = \langle |\delta_\ell u|^p \rangle, \qquad \ell \ll 1,
\end{equation}
which are the statistical $p$-moments of the increments $\delta_l u$ and contain much information of the motion of the fluid, depend only on the scale parameter $\ell$.
In \eqref{StructureFunctionPhysical}, $|\cdot |$ should be understood as some vector norm.
Kolmogorov deduced that the second moment should satisfy 
\begin{equation}\label{Kolmogorov2nd}
    S_2(\ell) \simeq \ell^{2/3}, 
\end{equation}
the so-called two-thirds law, when $\ell$ is in the inertial range. 
When $p$ is odd, it is usual to work with the longitudinal structure functions 
\begin{equation}
 S_p^{||}(\ell) = \Big\langle \left( \delta_\ell u  \cdot \frac{\ell}{|\ell|} \right)^p  \Big\rangle ,
\end{equation}
which account for possible cancellations. 
In fact, for the  third-order longitudinal structure function, Kolmogorov derived the celebrated four-fifths law, 
which states that in the limit of zero viscosity we have
\begin{equation}\label{Kolmogorov3rd}
    S_3^{||}(\ell) = -\frac45 \varepsilon\, \ell, \qquad \ell \ll 1,
\end{equation}
where $\epsilon$, the mean energy dissipation per unit mass, is assumed to be finite and non-zero.
In general, for higher order moments, with the additional assumption that the flow is statistically self-similar in small scales, he deduced
\begin{equation}\label{KolmogorovHigher}
    S_p(\ell) \simeq \ell^{p/3}, \qquad  \forall p \geq 4.
\end{equation}
While experimental data match \eqref{Kolmogorov2nd} and \eqref{Kolmogorov3rd}, they deviate considerably from \eqref{KolmogorovHigher}, which causes Kolmogorov's self-similar hypothesis being questioned. This departure from K41 and its self-similarity  is called intermittency, which is typically quantified using the flatness
\begin{equation}\label{FlatnessPhysical1}
    F(\ell) =  \frac{S_4(\ell)}{S_2(\ell)^2} = \frac{\langle |\delta_\ell u|^4 \rangle}{\langle |\delta_\ell u|^2 \rangle^2} , \qquad 0 < \ell \ll  1.
\end{equation}
Indeed, based on \eqref{KolmogorovHigher}, $F(\ell)$ should be constant, but $F(\ell) \simeq \ell^{-0.1}$ has been observed in experiments with turbulent flows like air and helium jets or in wind tunnels \cite{ChevillardEtAl2012}. In this context, when the flatness $F(\ell)$ grows without bound when $\ell \ll 1$, the flow is said to be intermittent in small scales.

The velocity increments \eqref{Velocity_Increments} are a useful way to describe the variations of the velocity, but not the only one: Fourier high-pass filters have also been used as an alternative in the physical literature (see \cite{BrunPumir,ChevillardEtAl} and \cite[Chapter 8.2]{Frisch}). Indeed, if $\widehat{u}$ denotes the Fourier transform, high-pass filters $u_{\geq N}(x) = \int_{|\xi| > N } \widehat{v}(\xi) e^{ix\xi}\, d\xi$ measure the effect of the frequencies of the signal that are higher than $N$, which have effects in scales smaller than $N^{-1}$.  Thus, the flatness
\begin{equation}\label{FlatnessPhysical2}
    F(N) = \frac{\langle |u_{ \geq N}|^4 \rangle}{ \langle |u_{\geq N}|^2 \rangle^2}, \qquad N \gg 1,
\end{equation}
is an analog to \eqref{FlatnessPhysical1}.
Both \eqref{FlatnessPhysical1} and \eqref{FlatnessPhysical2} have a probabilistic origin motivated by 
the concept of kurtosis. If~$X$ is a random variable of expected value $\mu$ and standard deviation $\sigma$, its kurtosis is defined by
\begin{equation}\label{Kurtosis}
    \operatorname{Kurt}[X] = \mathbb E \left[\left( \frac{X-\mu}{\sigma}\right)^4\right] = \frac{\mathbb E [(X-\mu)^4] }{ \mathbb E[(X-\mu)^2]^2},
\end{equation}
and it measures the ``tailedness" of its probability distribution \cite{Westfall}. This means that a large kurtosis represents a large probability of having outliers. In view of \eqref{FlatnessPhysical1}, the flatness of a turbulent flow is the kurtosis of $\delta_\ell u$, so a diverging flatness represents a large probability of outliers of the velocity increments at scale $\ell$. This inspires the name of \textit{intermittency at small scales}. A similar argument goes for the high-pass filter flatness \eqref{FlatnessPhysical2}. This probabilistic interpretation also caused the identification of intermittency with non-Gaussianity. Indeed, an intermittent fluid with growing flatness cannot have normally distributed increments since normal distributions of any mean and variance have a constant kurtosis equal to 3.

One could expect that the two approaches \eqref{FlatnessPhysical1} and \eqref{FlatnessPhysical2} give similar results after identifying the small scale parameters $\ell \simeq N^{-1}$. However, experimental and numerical data have shown that this may not be the case \cite{BrunPumir,ChevillardEtAl}. To clarify this analytically, in \cite{BoritchevEceizabarrenaVilacaDaRocha} we computed the two flatnesses for the classical Riemann's non-differentiable function (shown to be connected to fluids in  \cite{Jaffard1996,DelaHozVega,BanicaVega2020}), but the two have the same logarithmic growth.  
In this paper we delve deeper and study \eqref{FlatnessPhysical1} and \eqref{FlatnessPhysical2} for the functions
\begin{equation}\label{RiemannGeneralization}
    R_s(x) = \sum_{n=1}^{\infty}{\frac{e^{2\pi i n^2 x}}{n^{2s}}}, \quad x\in [0,1]. \qquad s > 1/2,
\end{equation}
These are generalizations of Riemann's non-differentiable function, which is recovered when $s=1$. With our computations, related to Rudin's $\Lambda(p)$-set problem for the set of squares, we prove that the two flatnesses are different overall, but also that they coincide for small $s$. We establish an analogy 
with the velocity of a turbulent fluid in which $s$ plays the role of viscosity, so that  the two flatnesses coincide in the range of $s$ that corresponds to the inertial (turbulent) range, but not in the dissipation (non-turbulent) range.  
Based on results of Jaffard \cite{Jaffard1}, we will also sketch the analytic reasons behind this, namely that the two flatnesses have the same leading order as long as the function is non-differentiable.

We also prove that all $R_s$ satisfy the multifractal formalism \eqref{MultifractalFormalism1}, based on \cite{Jaffard1996}. Indeed, this work is inspired by the mathematical adaptation of the multifractal formalism
\begin{equation}\label{MultifractalFormalism1}
    \zeta_p = \inf_{\alpha} \left\{ \alpha p + 3 - d(\alpha) \right\}, 
    \qquad \text{ or } \qquad 
    d(\alpha) = \inf_p \left\{ \alpha p + 3 - \zeta_p \right\},
\end{equation}
one of the central pieces of the multifractal models. These models were introduced by Frisch and Parisi \cite{FrischParisi}, \cite[Section 8.5]{Frisch} as a correction to \eqref{KolmogorovHigher} from K41 with the aim to capture the phenomenon of intermittency as observed in experiments. For \eqref{MultifractalFormalism1}, it is assumed that structure functions \eqref{StructureFunctionPhysical} follow a power law $S_p(\ell) \simeq \ell^{\zeta_p}$, and $d(\alpha) = \dim_{\mathcal H}D_\alpha$ is the spectrum of singularities of the velocity $u$, where $D_\alpha = \{ x : |u(x+\ell) - u(x)| \simeq \ell^\alpha \}$ is the set of points where $u$ is $\alpha$-H\"older regular 
and $\operatorname{dim}_{\mathcal H}$ denotes the Hausdorff dimension.
However, the derivation of \eqref{MultifractalFormalism1} was heuristic, which motivated its rigorous study in the setting of functions \cite{DaubechiesLagarias,Eyink,Jaffard1992,Jaffard1996,Jaffard1,Jaffard2,BarralSeuret}.
Our use of $R_s$ comes from Jaffard proving that Riemann's non-differentiable function $R_1$ satisfies the multifractal formalism \cite{Jaffard1996}. Riemann's function also appears as the trajectory of polygonal vortex filaments driven by the binormal flow \cite{DelaHozVega,BanicaVega2020}, and has thus been geometrically studied \cite{ChamizoCordoba1999,Eceizabarrena2019,Eceizabarrena1,Eceizabarrena2020}. Further generalizations of $R_s$ have also been studied \cite{ChamizoCordoba1999,ChamizoUbis2007,ChamizoUbis2014}.

According to all this, the functions $R_s$ capture some of the properties observed in experimental turbulence. 
However, 
while the chaotic structure of turbulence is usually studied from a statistical-probabilistic point of view, the functions $R_s$ are completely deterministic. 
Therefore, their power as models for turbulent phenomena is limited, and we should not expect them to capture all properties of turbulence. 
In this sense, it would be very interesting to develop a rigorous mathematical study of properties that we treat in this article, such as the flatness and the multifractal formalism, for other models that have been proposed for turbulence and intermittency with a random component.

\section{Results and discussion}\label{SECTION_Definitions_And_Results}

\subsection{Adaptation to the analytical setting}

Let us adapt the concepts presented in the introduction to the setting of functions, following Jaffard \cite{Jaffard1996,Jaffard1,Jaffard2}.
Let $\mathbb T = \mathbb R / \mathbb Z$ be the torus, $a_n \in \mathbb{C}$ and $f:\mathbb{T} \to \mathbb{C}$ a periodic function with Fourier series $f(x) =\sum_{n\in\mathbb{Z}}{a_n\,e^{2\pi i n x}}$. 
The usual way \cite[p. 946]{Jaffard1} to adapt structure functions from \eqref{StructureFunctionPhysical} to $f$ is 
\begin{equation}\label{StructureFunction}
    S_{f,p}(\ell) = \int_{\T}{|f(x+\ell)-f(x)|^p\,dx}, \qquad \ell \in [0,1], \qquad p \geq 0,
\end{equation} and high-pass filters from \eqref{FlatnessPhysical2} as 
\begin{equation}\label{HighPassFilterDefinition}
    f_{\geq N}(x) = \sum_{|n| \geq N}{a_n e^{2\pi i n x}}, \qquad N \in \mathbb N.
\end{equation}
To avoid $f_{\geq N} \equiv 0$, we assume that $f$ has infinitely many nonzero $a_n$ coefficients.
Low-pass filters are defined analogously by $ f_{\leq N}(x) = \sum_{|n| \leq N}{a_n e^{2\pi i n x}}$. 

\begin{rmk}
In the statistical analysis of the velocity field $u$ of a fluid, 
assuming homogeneity means that its statistical properties do not depend on the position. 
Thus, the statistical average in the definition of the structure functions in \eqref{StructureFunctionPhysical} can be interpreted as an average in space.
This motivates the definition in \eqref{StructureFunction}.
\end{rmk}

In the same way that the deviation of $S_4(\ell)$ from Kolmogorov's theory is measured by the flatness \eqref{FlatnessPhysical1}, the generalized flatness
\begin{equation}
    F_p(\ell) = \frac{S_p(\ell)}{S_2(\ell)^{p/2}} \qquad \qquad \text{ coming from } \quad \mathbb E \left[\left(\frac{|X-\mu|}{\sigma}\right)^p\right] = \frac{\mathbb E [|X-\mu|^p]}{\mathbb E [(X-\mu)^2]^{p/2}}
\end{equation}
is the analogue measure corresponding to $S_p(\ell)$, see \cite[(8.6)]{Frisch}.
Moreover, the generalized kurtosis on the right hand-side is an alternative measure of the thickness of the tails of a probability distribution. Thus, we adapt the flatness from \eqref{FlatnessPhysical1} and \eqref{FlatnessPhysical2} as follows.

\begin{defn}[Flatness in the sense of high-pass filters]\label{DefinitionFlatnessHP}
Let $2 < p < \infty$ and $N > 1$. We define the $p$-flatness
of $f$ at scale $N^{-1}$ in the sense of high-pass filters as
\begin{equation}\label{pqFlatness_HP}
    	F_{f,p}(N) = \frac{ \lVert f_{\geq N} \rVert_{L^p(\mathbb{T})}^p }{ \lVert f_{\geq N} \rVert_{L^2(\mathbb{T})}^p }. 
\end{equation}
\end{defn}

\begin{defn}[Flatness in the sense of structure functions]\label{DefinitionFlatnessSF}
Let $2 < p < \infty$ and $0 < \ell < 1$. We define the $p$-flatness
of $f$ at scale $\ell$ in the sense of  structure functions as
\begin{equation}\label{pqFlatness_SF}
    	G_{f,p}(\ell) = \frac{ S_{f,p}(\ell) }{S_{f,2}(\ell)^{p/2} }. 
\end{equation}
\end{defn}

\subsection{Results}

In \cite{BoritchevEceizabarrenaVilacaDaRocha},  filaments, 
we proved that
\begin{equation}
    F_{R_1,4}(N) \simeq \log N, \quad (N \gg 1)  \qquad \qquad   G_{R_1,4}(\ell) \simeq \log \ell^{-1}, \quad (\ell \ll 1)
\end{equation}
where $R_1$ is Riemann's non-differentiable function. 
Here we study the generalizations $R_s$ defined in \eqref{RiemannGeneralization} to compare the two flatness $F_{f,p}(N)$ and $G_{f,p}(\ell)$ in a more general situation. 
These functions are absolutely convergent when $s>1/2$, but can be shown to be in $L^p(\mathbb T)$ for lower values of $s$, so for $p\geq2$ we define the threshold 
\begin{equation}\label{defsstar}
 s_p^\star=
 \left\{  \begin{array}{ll}
        1/4, & p \leq 4, \\
        1/2-1/p, & p>4.
        \end{array}
    \right.
\end{equation}

\begin{thmA}\label{Theorem_FlatnessHP}
Let $p\geq2$, $s>s_p^\star$ defined in \eqref{defsstar} and $R_s$ defined in \eqref{RiemannGeneralization}. 
Then, as $N \to \infty$,
\begin{equation}
    F_{R_s,p}(N) \simeq \left\{  \begin{array}{ll}
        1, & p < 4, \\
        \log N, & p=4, \\
        N^{p/4 - 1}, & p > 4.
        \end{array}
    \right.
\end{equation}
\end{thmA}

\begin{thmA}\label{Theorem_FlatnessSF}
Let $p\geq2$, $s>s_p^\star$ defined in \eqref{defsstar} and $R_s$ defined in \eqref{RiemannGeneralization}. Then, as $\ell \to 0$,
\begin{itemize}
    \item when $s < 5/4$, 
    \begin{equation}
    G_{R_s,p}(\ell) \simeq \left\{  \begin{array}{ll}
        1, & p < 4, \\
        \log (\ell^{-1}), & p=4, \\
        \ell^{-(p/4 - 1)}, & p > 4.
    \end{array}
    \right.
\end{equation}

    \item when $s=5/4$, 
\begin{equation}
    G_{R_s,p}(\ell) \simeq \left\{  \begin{array}{ll}
        1, & p < 4, \\
        \ell^{-(p/4 - 1)}\, (\log (\ell^{-1}))^{-p/2}, & p > 4,
    \end{array}
    \right. 
\end{equation}
and when $p=4$ we have $1 \lesssim G_{R_s,4}(\ell) \lesssim \log(\ell^{-1})$.

    \item when $5/4 < s < 3/2$, 
\begin{equation}
    G_{R_s,p}(\ell) \simeq \left\{  \begin{array}{ll}
        1, & p <  2/(3-2s), \\
        \ell^{-(3p/2 - 1 - ps)}, & p > 2/(3-2s),  
        \end{array}
    \right. 
\end{equation}
and when $p=2/(3-2s)$ we have $1 \lesssim G_{R_s, p}(\ell) \lesssim (\log(\ell^{-1}))^{p/2}$.

    \item when $s \geq 3/2$, we have $G_{R_s,p}(\ell) \simeq 1$ for all $p>0$.

\end{itemize}
\end{thmA}

Theorem~\ref{Theorem_FlatnessHP} follows from the behavior of $\lVert (R_s)_{\geq N} \rVert_{L^p(\mathbb{T})}$, which we compute in Theorem~\ref{TheoremHighPassFilters}.  In the same way, Theorem~\ref{Theorem_FlatnessSF} is a consequence of the behavior of the structure functions $S_p(\ell)$, which we give in Theorem~\ref{TheoremStructureFunctions}. Both Theorems~\ref{TheoremHighPassFilters} and \ref{TheoremStructureFunctions} are a consequence of estimates for the Littlewood-Paley blocks  and low-pass filters of $R_s$ that we compute in Proposition~\ref{PropDeltak} and in Theorem~\ref{TheoremNewZalcwasser}, respectively.  They could be of independent interest; for instance, they are related to Rudin's $\Lambda(p)$-set problem for the set of squares \cite{Rudin,Bourgain1989}, which is relevant among others in the study of the nonlinear Schr\"odinger equation \cite{Bourgain1993Part1}.

\begin{corA}\label{Corollary_Discordance}
The flatness in the sense of high-pass filters (Definition~\ref{DefinitionFlatnessHP}) and the flatness in the sense of structure functions (Definition~\ref{DefinitionFlatnessSF}) are not equivalent in general, not even for the standard case of $p=4$, as shown by $R_s$ with $s > 5/4$. 
\end{corA}

We also prove the multifractal formalism for the functions $R_s$. Let us briefly discuss how \eqref{MultifractalFormalism1} is adapted to the setting of functions; for a detailed account we refer the reader to \cite{Jaffard1,Jaffard2}. A function $f$ is said to be $\alpha$-H\"older at the point $x_0$, and written $f\in C^\alpha(x_0)$, if there is a polynomial $P$ of degree at most $\alpha$ such that $\left| f(x_0+h) - P(h)  \right| \leq C\, h^\alpha$, and the H\"older exponent of $f$ at $x_0$ is $\alpha_f(x_0) = \sup\{ \alpha \mid f \in C^\alpha(x_0)\}$. The $\alpha$-H\"older regularity sets are $D_\alpha = \{ x \mid \alpha_f(x) = \alpha \}$, and $d(\alpha) = \dim_\mathcal{H} D_\alpha$ is the Hausdorff dimension of these sets, called the spectrum of singularities of $f$. On the other hand, the exponent $\zeta_p$ can be defined in several ways \cite{Jaffard1} which we discuss later in \eqref{Exponents}.

\begin{thmA}\label{Theorem_MultifractalFormalism}
Let $s > 1/2$, $p>0$ and $\eta_s(p) =  \sup \{ \sigma \mid R_s \in B^{\sigma/p}_{p,\infty}  \}$, where $B^\sigma_{p,q}$ denote Besov spaces.
Then, the functions $R_s$ defined in \eqref{RiemannGeneralization} satisfy the multifractal formalism in dimension 1, that is, \begin{equation}
    d_s(\alpha) = \inf_{p>0} \{ \alpha p - \eta_s(p) + 1  \},
\end{equation}
for the values of $\alpha$ where $d_s(\alpha)$ is increasing. 
\end{thmA}

In \cite[Corollary 2]{Jaffard1996}, Jaffard proved that
\begin{equation}\label{SpectrumOfSingularitiesOfRs}
    d_s(\alpha) = \left\{  \begin{array}{ll}
        4\alpha - 4s + 2, & \alpha \in [s-1/2,s-1/4], \\
        0, & \alpha = 2s - 1/2, \\
        - \infty, & \text{otherwise},
    \end{array} \right. \qquad \forall s > 1/2,
\end{equation}
where by convention the dimension is $-\infty$ when the set is empty.
Thus, in order to prove Theorem~\ref{Theorem_MultifractalFormalism}, we compute $\eta_s(p)$:
\begin{propA}\label{Proposition_MultifractalEta}
Let $s>1/2$, $p >0$ and $\eta_s(p) =  \sup \{ \sigma \mid R_s \in B^{\sigma/p}_{p,\infty}  \}$. Then, 
\begin{equation}
\eta_s(p) = \left\{
    \begin{array}{ll}
        p(s-1/4), & p \leq 4, \\
        1 + p(s - 1/2), & p > 4.
    \end{array}
    \right.
\end{equation}
\end{propA}

The proof of Proposition~\ref{Proposition_MultifractalEta} is directly linked to the proof of Theorem~\ref{Theorem_FlatnessHP}, since $\eta_s(p)$ is related to the Littlewood-Paley decomposition of $R_s$ and, thus, to the $L^p$ estimates of the high-pass filters $\lVert (R_s)_{\geq N} \rVert_{L^p(\mathbb T)}$ which we compute in Theorem~\ref{TheoremHighPassFilters}.

\subsection{Discussion of results}

Theorems~\ref{Theorem_FlatnessHP} and \ref{Theorem_FlatnessSF} show that the two flatnesses $F_{s,p}$ and $G_{s,p}$ do not always coincide. A visual representation of this discrepancy is shown in Figure~\ref{Figure_LevelCurves}. While the flatness of high-pass filters (Figure~\ref{Figure_LevelCurvesHP}) does not depend on $s$, structure functions produce a handful of regions (Figure~\ref{Figure_LevelCurvesSF}).

\begin{figure}[h!]
  \makebox[\textwidth][c]{ %
   \begin{subfigure}{0.60\linewidth}
    \includegraphics[width=\linewidth]{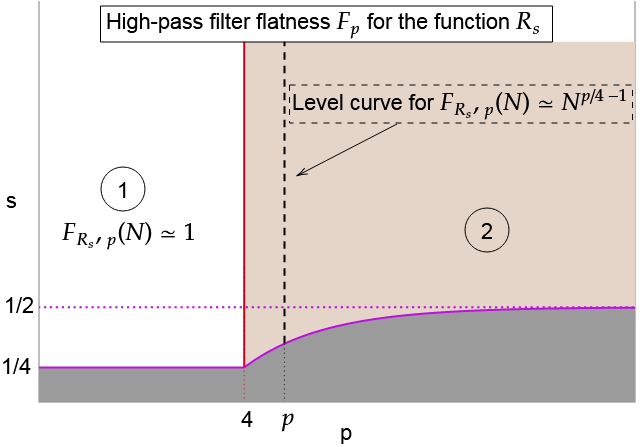}
    \caption{In the sense of high-pass filters (Theorem~\ref{Theorem_FlatnessHP})}
    \label{Figure_LevelCurvesHP}
    \end{subfigure}
    \begin{subfigure}{0.60\linewidth}
    \includegraphics[width=\linewidth]{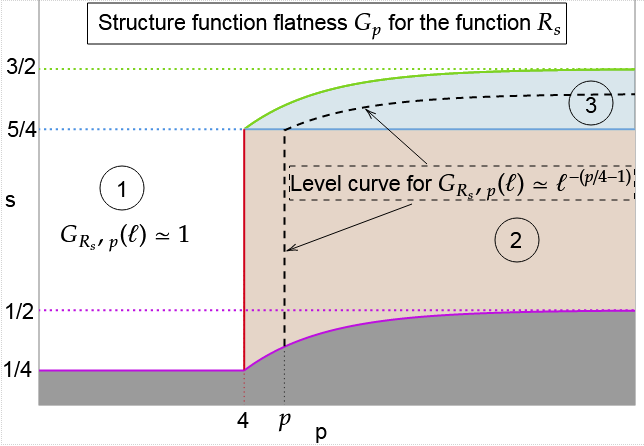}
    \caption{In the sense of structure functions (Theorem~\ref{Theorem_FlatnessSF})}
    \label{Figure_LevelCurvesSF}
    \end{subfigure}
  }
  \caption{Comparison of flatnesses $F_{R_s,p}$ and $G_{R_s,p}$. Dashed lines are level curves for the powers of the flatness.}
  \label{Figure_LevelCurves}
\end{figure}

Up to the behavior in the boundaries, these regions are:
\begin{itemize}
    \item Region $\circled{1}$. 
    The non-intermittent region, where the flatness is constant. The regions coincide only when $p<4$. 
    \item Region $\circled{2}$. 
    We have intermittent, divergent power laws $F_{R_s,p}(N) \simeq N^{p/4-1}$ and $G_{R_s,p}(\ell) \simeq \ell^{-(p/4 - 1)}$. The regions match when $p>4$ and $s<5/4$.
    \item Region $\circled{3}$. Exclusive of structure functions, 
    where $G_{R_s,p}(\ell) \simeq \ell^{-(3p/2 - 1 - ps)}$. For fixed $p$, intermittency is weaker than in region $\circled{2}$, since 
    \[   \ell^{-(3p/2-1-ps)} \ll \ell^{-(p/4 - 1)}  \qquad \text{ when } s>5/4. \]
\end{itemize}

As we said in the introduction, these coincidences and discrepancies have also been observed in experimental turbulence \cite{BrunPumir,ChevillardEtAl}. Indeed, we can establish an analogy between our results for $R_s$ in Theorems~\ref{Theorem_FlatnessHP} and \ref{Theorem_FlatnessSF}  and  the effect of viscosity in a turbulent flow, where the parameter $s$ plays the role of viscosity. For that, in Figure~\ref{Figure_Plots} we plot the complex functions $R_s$ \eqref{RiemannGeneralization} on the left column and the classical real-valued $\operatorname{Im}R_s$ on the right column. Different $s$ produce quite different situations, larger $s$ showing less oscillations and more regularity. When $s<5/4$, Figures~\ref{Figure_B} and \ref{Figure_D} are similar to the typical irregular velocity profile of turbulent flows in the inertial range, where the effect of viscosity is small. When $s \in (5/4,3/2)$, Figure~\ref{Figure_F} is smoother, representing a transition to the dissipation range where viscosity starts having an effect;  we further discuss this range in the end of the section. Finally, when $s > 3/2$,  Figure~\ref{Figure_H} represents the dissipation range where viscosity dominates and completely smooths out the velocity. Thus, the two flatnesses coincide only for $s$ corresponding to the inertial range, so this analogy suggests that in experimental observations the two approaches should be interchangeable in the inertial range, but not in the dissipation range.

\begin{figure}

\begin{subfigure}{0.38\linewidth}
\includegraphics[width=\linewidth]{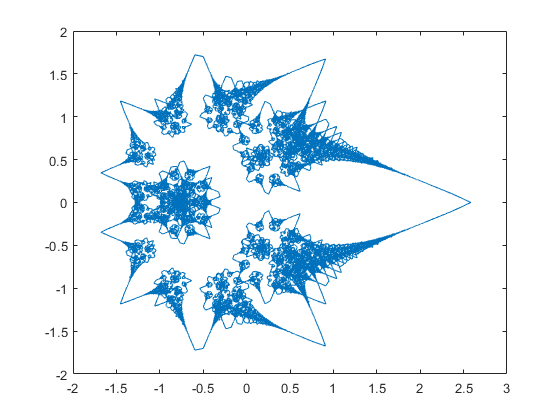}
\caption{$s=0.75$}
\label{Figure_A}
\end{subfigure}
\begin{subfigure}{0.38\linewidth}
\includegraphics[width=\linewidth]{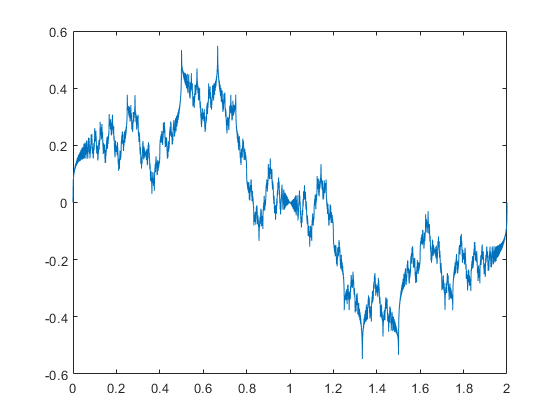}
\caption{$s=0.75$}
\label{Figure_B}
\end{subfigure}

\begin{subfigure}{0.38\linewidth}
\includegraphics[width=\linewidth]{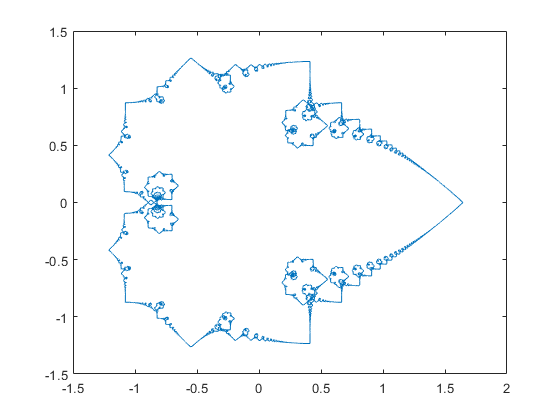}
\caption{$s=1$}
\label{Figure_C}
\end{subfigure}
\begin{subfigure}{0.38\linewidth}
\includegraphics[width=\linewidth]{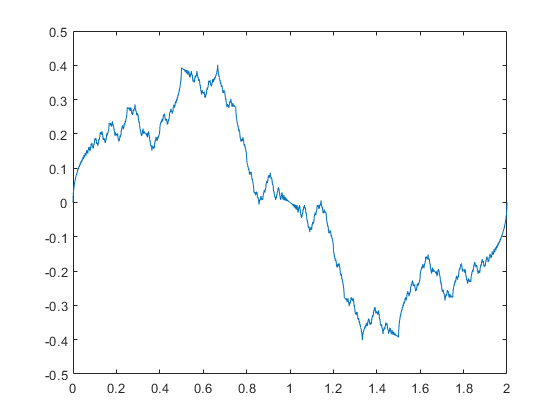}
\caption{$s=1$}
\label{Figure_D}
\end{subfigure}

\begin{subfigure}{0.38\linewidth}
\includegraphics[width=\linewidth]{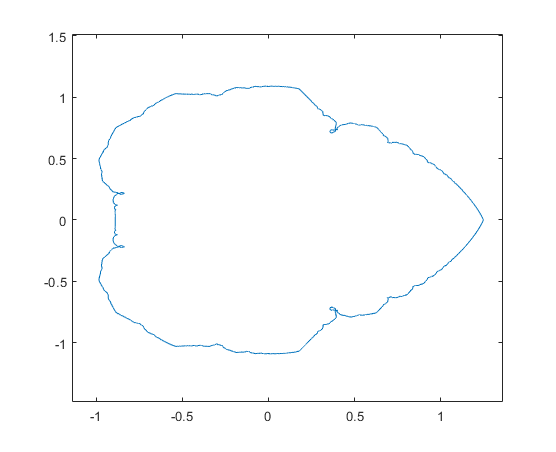}
\caption{$s=1.4$}
\label{Figure_E}
\end{subfigure}
\begin{subfigure}{0.38\linewidth}
\includegraphics[width=\linewidth]{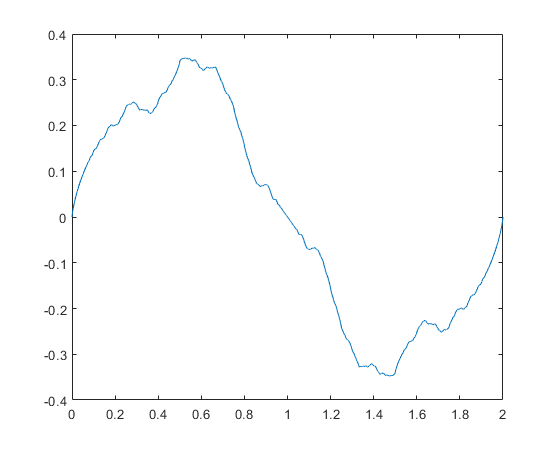}
\caption{$s=1.4$}
\label{Figure_F}
\end{subfigure}

\begin{subfigure}{0.38\linewidth}
\includegraphics[width=\linewidth]{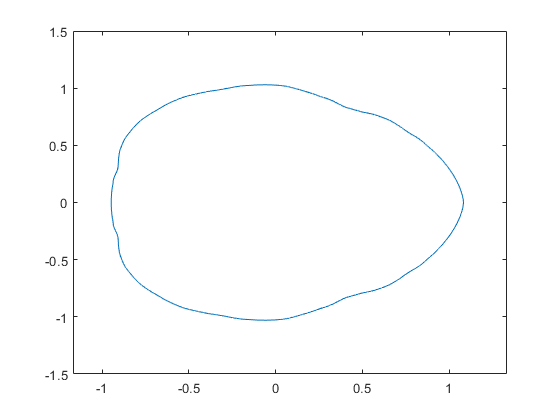}
\caption{$s=2$}
\label{Figure_G}
\end{subfigure}
\begin{subfigure}{0.38\linewidth}
\includegraphics[width=\linewidth]{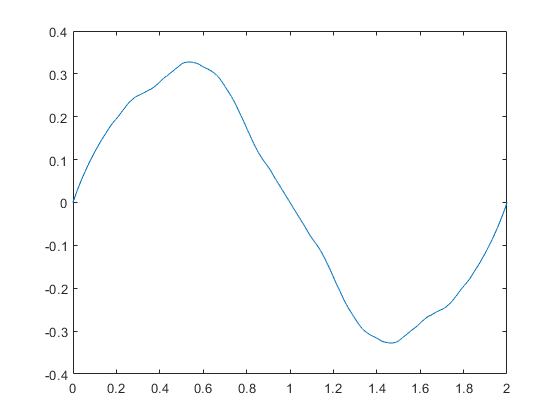}
\caption{$s=2$}
\label{Figure_H}
\end{subfigure}

\caption{Left column: the sets $R_s([0,1]) \subset \mathbb C$ for different values of $s$. Right column: corresponding graphs of $\operatorname{Im} R_s(x/2) / \pi $ for $x\in[0,2]$.}
\label{Figure_Plots}
\end{figure}

This analogy can be further supported analytically. The main difference between the inertial and the dissipation range is that the flow is extremely irregular in the former, but smooth in the latter. The same goes for $R_s$, since Jaffard's result in \eqref{SpectrumOfSingularitiesOfRs} implies that it is almost nowhere differentiable when $s<5/4$, but almost everywhere differentiable when $5/4 < s < 3/2$ and everywhere differentiable when $s\geq 3/2$. Thus, flatnesses 
coincide only in the irregular inertial range, when $R_s$ is almost nowhere differentiable.
Up to lower order corrections, this is a general fact that is connected to the exponent $\zeta_p$ in the multifractal formalism \eqref{MultifractalFormalism1}. According to \cite[Proposition 3.1]{Jaffard1}, $\zeta_p$ may be equivalently defined in terms of several Sobolev-like spaces, such as
\begin{equation}\label{Exponents}
\begin{split}
    \zeta_p & = \sup\{ \sigma \mid f \in B^{\sigma/p}_{p,\infty} \} = \sup\{ \sigma \mid f \in H^{\sigma/p,p} \} \\
    & = \sup\{ \sigma \mid f \in L^{\sigma/p,p} \} = \sup\{ \sigma \mid f \in W^{\sigma/p,p} \}, 
\end{split}
\qquad \text{ for } p > 1,
\end{equation}
where $B^\sigma_{p,q}$ are Besov spaces, $H^{\sigma,p}$ are Nikolskii spaces, $L^{\sigma,p}$ are Bessel potential spaces and $W^{\sigma,p}$ are Sobolev-Slobodeckij spaces. The first two are directly connected to Fourier filters and structure functions respectively, as we explain in the following lines.

Besov spaces in $\mathbb T$ are based in the Littlewood-Paley decomposition of a function,
\begin{equation}\label{Besov}
    f(x) = \sum_{n\in \mathbb Z} \widehat{f}_n \, e^{2\pi i n x}  = \sum_{k=0}^\infty \left(  \sum_{|n|=2^k}^{2^{k+1}-1} \widehat{f}_n \, e^{2\pi i n x} \right) = \sum_{k=0}^\infty P_k f(x),
\end{equation}
since $f \in B^\sigma_{p,q}$ if the sequence $\left( 2^{k\sigma} \lVert  P_k f \rVert_p \right)_k$ is in $\ell^q$. That means that 
\begin{equation}
    f \in B^{\sigma/p}_{p,\infty} \quad \Longleftrightarrow \quad \sup_k 2^{k\sigma} \lVert  P_k f \rVert_p^p \leq C.
\end{equation}
Denoting $N=2^k$, we see that $P_k f = f_{\simeq N}$ is a band-pass filter of $f$. From \eqref{Besov} and the definition of $\zeta_p$ in \eqref{Exponents}, formally we have $\lVert f_{\simeq N} \rVert_p^p \simeq N^{-\zeta_p}$, where lower order corrections like logarithms are not accounted for. Also, in many situations like for $R_s$, $\lVert f_{\simeq N} \rVert_p \simeq \lVert f_{\geq N} \rVert_p$ holds.

Regarding Nikolskii spaces $H^{\sigma,p}$, the original purpose of $\zeta_p$ as the power law of structure functions \eqref{StructureFunction} is recovered. Indeed, if $\sigma = m + \sigma'$ with $m \in \mathbb N$ and $\sigma' \in [0,1)$, 
\begin{equation}\label{Nikolskii}
    f \in H^{\sigma,p} \quad \Leftrightarrow \quad f \in L^p \quad \text{ and } \quad \int \frac{|D^mf(x + \ell ) - D^mf(x)|^p}{\ell^{\sigma' p}} \, dx \leq C,
\end{equation}
where $D^m f$ is the $m$-th derivative of $f$. When $\sigma < 1$, there is a clear connection with $S_{f,p}(\ell)$, since if $\zeta_p < p$ the definition of $\zeta_p$ \eqref{Exponents} formally implies $S_{f,p}(\ell) \simeq \ell^{\zeta_p}$. 

Consequently, joining these two interpretations of $\zeta_p$, structure functions and band-pass filters follow the same power law up to lower order corrections if $\zeta_p < p$ (roughly, one may think of this situation as $f' \notin L^p$).
Thus, when $f$ is not regular enough, like the velocity of the flow in the inertial range, one should in general expect that $F_{f,p}$ and $G_{f,p}$ have a similar behavior. On the other hand, when the function is regular enough so that $\zeta_p > p$, the Besov space characterization still gives $\lVert f_{\simeq N} \rVert_p^p \simeq N^{-\zeta_p}$, but for the structure functions we may formally write
\begin{equation}
    S_{f,p}(\ell) = \int |f(x + \ell) - f(x)|^p \, dx \simeq \int |f'(x)|^p \, \ell^p \,dx = \ell^p \lVert f' \rVert_p^p \simeq \ell^p \gg N^{-\zeta_p},
\end{equation}
if $N^{-1} \simeq \ell$. The reason behind this is that, according to the definition $H^{\sigma,p}$ in \eqref{Nikolskii}, $\zeta_p$ describes no longer the structure functions of $f$ but of its derivatives.
Thus, we should expect the two flatnesses to coincide in irregular situations like in the inertial range of turbulence, but not when the element under consideration becomes regular, like in the dissipation range by the effect of viscosity.

We may as well draw an interesting analogy with the debate on the correctness of the multifractal formalism \eqref{MultifractalFormalism1} in the dissipation range.
Indeed, according to Theorem~\ref{Theorem_MultifractalFormalism}, the function $R_s$ satisfies the multifractal formalism with the exponent $\zeta_p$ defined as in \eqref{Exponents}, 
that is, with the exponent corresponding to high-pass filters.
But what about the exponent $\xi_p$ coming from the structure functions $S_p(\ell) \simeq \ell^{\xi_p}$, as proposed in the original multifractal formalism?

As we prove in  Theorems~\ref{Theorem_FlatnessSF} and \ref{TheoremStructureFunctions},
we have $\xi_p = \zeta_p$ for $s < 5/4$, 
in which case the multifractal formalism is satisfied also with structure functions.
However, when $s\geq 5/4$ the exponents $\xi_p$ and $\zeta_p$ are no longer equal,
and since by Theorem~\ref{Theorem_MultifractalFormalism} the multifractal formalism still holds with $\zeta_p$, it can not be satisfied with $\xi_p$ from structure functions. 
We can, though, separate two regimes:
\begin{itemize}

    \item if $s > 3/2$, then $R_s$ is differentiable, so from the point of view of turbulence there is not much left to understand in terms of irregularity. As suggested before, this is analog to the small scales where viscosity clearly dominates and the velocity of the fluid becomes smooth.
    In analytic terms, even if the structure functions of $R_s$ do not satisfy the multifractal formalism, it is in fact hidden in the derivatives of $R_s$. 
    Indeed, since $s-1 > 1/2$, then $R_s' = R_{s-1}$ is absolutely convergent and belongs to all $L^p$ spaces. 
    Thus, if $3/2 < s < 5/2$,  by \eqref{Nikolskii} $\zeta_p$ captures the behavior of the structure functions of $R'_s$. In other words, it is the structure functions of $R_s'$ which satisfy the multifractal formalism.
    For $s > 5/2$, the same works with higher order derivatives.

    \item 
    But if $5/4 < s < 3/2$, $R_s' = R_{s-1}$ is not absolutely convergent, it is not in all $L^p$ and $\xi_p$ is not defined for all $p$.
    In other words, there is no evident way to relate $\zeta_p$ with structure functions, so there is no direct way to justify the multifractal formalism from the approach of structure functions. 
    
    In turbulence, this is analog to the idea that the dissipation scale is not  a single scale, but rather a whole range of scales of transition from the inertial range to the dissipative range, where the multifractal formalism may not be accurate enough. 
    Corrections to the multifractal formalism in this dissipative transtion have been proposed in the physical literature (see the review in \cite[Chapter 8.5.5.]{Frisch})
    It would be a very interesting task to further study this transition regime from a rigorous mathematical perspective. 
    
\end{itemize}

\subsection{Structure of the article}

We briefly describe the structure of this article. 

\begin{itemize}
    \item In Section~\ref{SECTION_AuxiliaryResults} we prove technical lemmas that we use throughout the paper.
    
    \item In Section~\ref{SECTION_LpNorms} we give $L^p$ estimates on the Littlewood-Paley blocks (Proposition~\ref{PropDeltak}) and low-pass filters (Theorem~\ref{TheoremNewZalcwasser}) of $R_s$. We use them to prove the results in Sections~\ref{SECTION_HighPassFilters} and \ref{SECTION_StructureFunctions}, but as mentioned, they could also be of independent interest. 
    
    \item In Section~\ref{SECTION_HighPassFilters} - Theorem~\ref{TheoremHighPassFilters}, we compute the behavior of the high-pass filters $\lVert (R_s)_{\geq N} \rVert_p$. For that, we use the Littlewood-Paley decomposition and Proposition~\ref{PropDeltak}. Theorem~\ref{Theorem_FlatnessHP} is a direct consequence of it.
    
    \item  In Section~\ref{SECTION_StructureFunctions} - Theorem~\ref{TheoremStructureFunctions}, we compute the behavior of the structure functions $S_{R_s,p}(\ell)$ using the low pass filters in Theorem~\ref{TheoremNewZalcwasser}. Theorem~\ref{Theorem_FlatnessSF} immediately follows. 
    
    \item In Section~\ref{SECTION_MultifractalFormalism} we prove Theorem~\ref{Theorem_MultifractalFormalism} and Proposition~\ref{Proposition_MultifractalEta} concerning the multifractal formalism  for $R_s$.
    
\end{itemize}

\subsection{Notation}
Throughout the paper, we adopt the following notation.

For inequalities:

\begin{itemize}
    \item We write $A \lesssim B$ when there is a constant $C >0$, independent of the parameters of the problem, such that $A \leq CB$.
    We write $A \simeq B$ when $A\lesssim B$ and $B \lesssim A$.
    The implicit constants in the former inequalities may change from line to line.
    If the constants of the former inequalities depend on a certain parameter, say $p$, then we write $\lesssim_p$ and $\simeq_p$.
\end{itemize}

For the limits of a sum:
\begin{itemize}
    \item If either of $A,B \in \mathbb R$ are not integers, we denote $\sum_{n=A}^B = \sum_{n \in [A,B] \cap \mathbb Z}$.
\end{itemize}

For $L^p$ norms and the exponential functions:
\begin{itemize}
    \item We write $\lVert \cdot \rVert_p$ instead of $\lVert \cdot \rVert_{L^p(\mathbb T)}$.
    
    \item  We denote $e_n(x) = e^{2\pi i n x}$.
    
    \item Many times, especially after Section~\ref{SECTION_LpNorms}, we will write $\left\lVert  \sum_{n=A}^{B} e^{2\pi i n^2 x}/n^{2s} \right\rVert_p$
    instead of $\left\lVert  \sum_{n=A}^{B} e_{n^2}/n^{2s} \right\rVert_p$.
\end{itemize}

For the Fourier transform:
\begin{itemize}
    \item We denote the Fourier transform $\mathcal F (f)$ by $\widehat{f}$ and the inverse Fourier transform $\mathcal F^{-1}(f)$ by $f^{\vee}$.
    \begin{equation}
        \widehat{f}(\xi) = \int_{\R} f(x) \, e^{-2\pi i \xi x} \, \mathrm dx, 
        \qquad \qquad 
        f^{\vee}(\xi) = \int_{\R} f(x) \, e^{2\pi i \xi x} \, \mathrm dx.
    \end{equation}
    
\end{itemize}

\section{Auxiliary results}\label{SECTION_AuxiliaryResults}

We begin with a lemma that allows us to extract power coefficients from the $L^p$ norms of localized Fourier series.

\begin{lem}\label{lemmaluis}
Let $p \geq 1$, $1\leq A<B$, $(a_n)\in\mathbb C^\mathbb N$ a sequence 
and $\gamma \geq 0$. Then, 
    \begin{equation} \label{gammapos}
          A^\gamma \, \Bigg\| \sum_{n=A}^{B} a_n e_n \Bigg\|_{p} \lesssim_{\gamma, \frac BA} \Bigg\| \sum_{n=A}^{B} n^\gamma a_n e_n \Bigg\|_{p}
         \lesssim_{\gamma, \frac BA} B^\gamma \, \Bigg\| \sum_{n=A}^{B} a_n e_n \Bigg\|_{p},
    \end{equation}
    and
    \begin{equation} \label{gammaneg}
         \frac{1}{B^\gamma} \, \Bigg\| \sum_{n=A}^{B} a_n e_n \Bigg\|_{p} \lesssim_{\gamma, \frac BA} \Bigg\| \sum_{n=A}^{B} \frac{a_n}{n^\gamma} \, e_n \Bigg\|_{p}
        \lesssim_{\gamma, \frac BA} \frac{1}{A^\gamma} \, \Bigg\| \sum_{n=A}^{B} a_n e_n \Bigg\|_{p}.
    \end{equation}
Moreover, if $\gamma \in \mathbb N$ or if $\gamma \geq 2$, then 
\begin{equation}
    \Bigg\| \sum_{n=A}^{B} n^\gamma a_n e_n \Bigg\|_{p}
         \lesssim_{\gamma} B^\gamma \, \Bigg\| \sum_{n=A}^{B} a_n e_n \Bigg\|_{p}.
\end{equation}
\end{lem}

\begin{proof}
Let $\gamma \geq 0$ and define
\begin{equation}
    f =\sum_{n=A}^{B} n^\gamma a_n e_n, \qquad \text{ and } \qquad g =\sum_{n=A}^{B} a_n e_n.
\end{equation}
Let us prove the right-hand side inequality in \eqref{gammapos} first.
Let $\chi = \chi_{A/B}\in \mathcal S(\R)$ be a cut-off function depending on $A/B$
such that
\begin{equation}\label{CutoffLemmaLuis}
\chi(\xi) =  \begin{cases}
1 \quad \text{ if } \quad A/B \leq \xi \leq 1,\\
0 \quad \text{ if } \quad \xi \leq A/(2B) \quad \text{or} \quad \xi \geq 2,
\end{cases}
\end{equation}
and define 
$m_\gamma(\xi) = \xi^\gamma \, \chi(\xi)$ for $\xi\in\R$.
Since $0 \notin \operatorname{supp}\chi$, we have $m_\gamma\in\mathcal S(\R)$.
Define also the scaled function
$m_{\gamma,B}(\xi) =m_\gamma\left(\xi / B \right)$ and write 
\begin{equation}
 f=\sum_{n=A}^{B}n^\gamma a_n e_n=\sum_{n=A}^{B}n^\gamma\chi\left(\frac{n}B\right) \,  a_n \, e_n
  =B^\gamma\sum_{n=A}^{B}m_{\gamma,B}(n)\, a_n \, e_n.
\end{equation}
Let us introduce the auxiliary function
\begin{equation}\label{AuxiliaryFunction}
 h_{\gamma,B} = \sum_{n \in \mathbb Z} m_{\gamma,B}(n) \, e_n = \sum_{n=A/2}^{2B}m_{\gamma,B}(n) \, e_n,
\end{equation}
so that $f = B^\gamma \left(g\ast h_{\gamma,B}\right)$, where $\ast$ denotes the convolution in $\mathbb{T}$.
Indeed, for $x\in\mathbb T$ we have
\begin{equation}
    \begin{split}
   g\ast h_{\gamma,B}(x)
   & = \int_{\mathbb T}g(y) \, h_{\gamma,B}(x-y)\mathrm dy =
   \sum_{j=A}^{B}\sum_{k=A/2}^{2B}a_jm_{\gamma,B}(k) \, e_k(x)\, \int_{\mathbb T} e^{2\pi i(j-k) y}\, \mathrm dy \\
   & = \sum_{j=A}^{B} a_j \, m_{\gamma,B}(j) \, e_j(x) 
   = B^{-\gamma} f(x).
    \end{split}
\end{equation}
Thus, by Young's inequality we get
\begin{equation}\label{Young}
   \left\| f \right\|_p\leq B^\gamma\left\| g \right\|_p\left\|  h_{\gamma,B} \right\|_1,
\end{equation}
so it suffices to bound $\left\|  h_{\gamma,B} \right\|_1$.
For that purpose, let us use Poisson's summation formula in \eqref{AuxiliaryFunction}, 
which we can use because the function $F_x(y) = m_{\gamma, B}(y)\, e^{2\pi i x y}$ is Schwartz for every $x \in \mathbb T$.
Therefore, 
\begin{equation}\label{PoissonSummationFormula1}
    h_{\gamma,B}(x) = \sum_{n \in \mathbb Z} F_x(n)  
    = \sum_{j \in \mathbb Z} \widehat {F_x} (j) 
    = \sum_{j \in \mathbb Z} \widehat{m_{\gamma, B}}(j-x)
    =\sum_{j \in \mathbb Z} (m_{\gamma, B})^{\vee}(x-j).
\end{equation}
Since $(m_{\gamma, B})^{\vee}(y) = \int_{\R} m_\gamma(\xi/B) \, e^{2\pi i \xi y} \, d\xi 
= B (m_\gamma)^{\vee}(B y)$, 
we get
\begin{equation}\label{PoissonSummationFormula2}
    \lVert h_{\gamma, B} \rVert_{1} \leq \sum_{j \in \mathbb Z} \int_0^1 \left|  (m_{\gamma, B})^{\vee}(x-j)  \right| dx =  \int_{\mathbb R} \left|  (m_{\gamma, B})^{\vee}(x)  \right| dx
    = \left\lVert (m_\gamma)^{\vee} \right\rVert_{L^1(\mathbb R)},
\end{equation}
which depends on $\gamma$ and $A/B$ because $m_\gamma$ depends on $\gamma$ and $\chi = \chi_{A/B}$.
Thus, from \eqref{Young} we get 
\begin{equation}\label{IntermediateStepLemmaLuis}
\left\| f \right\|_p
\leq B^\gamma \, \left\lVert (m_\gamma)^{\vee} \right\rVert_{L^1(\mathbb R)} \, \left\| g \right\|_p
\lesssim_{\gamma,\frac AB} \, B^\gamma \,  \left\| g \right\|_p
\end{equation}
as desired.
Now, the left-hand side inequality in \eqref{gammaneg} follows immediately
because defining $\tilde a_n = a_n / n^\gamma $, 
we can use \eqref{IntermediateStepLemmaLuis} so that
\begin{equation}\label{EasyCaseInLemma}
 \left\| \sum_{n=A}^{B}a_n \, e_n \right\|_{p}=
 \left\| \sum_{n=A}^{B}n^{\gamma} \tilde{a}_n e_n \right\|_{p}
 \lesssim_{\gamma, \frac AB} B^{\gamma} \left\| \sum_{n=A}^{B} \tilde{a}_n \, e_n \right\|_{p}
 = B^{\gamma} \left\| \sum_{n=A}^{B} \frac{a_n}{n^\gamma} \,   e_n \right\|_{p}.
\end{equation}

The right hand-side inequality in \eqref{gammaneg} is proved similarly with an alternative cutoff function $\widetilde{\chi} = \widetilde\chi_{B/A} \in \mathcal S(\mathbb R)$ such that 
\begin{equation}\label{CutoffLemmaLuis2}
\widetilde\chi(x) =  \begin{cases}
1 \quad \text{ if } \quad 1 \leq x \leq B/A,\\
0 \quad \text{ if } \quad x \leq 1/2 \quad \text{or} \quad x\geq 2B/A,
\end{cases}
\end{equation}
Define the Schwartz functions $ \widetilde m_\gamma(\xi) = \xi^{-\gamma}\, \widetilde\chi(\xi)$ and $\widetilde m_{\gamma,A}(\xi) = \widetilde m_\gamma( \xi / A)$ for $\xi\in\R$, so
\begin{equation}
    f = \sum_{n=A}^B \frac{a_n}{n^\gamma} \,   e_n  = \sum_{n=A}^B \frac{1}{n^\gamma}\,  \tilde\chi \left(\frac{n}{A}\right) \, a_n \, e_n = \frac{1}{A^\gamma} \,  \sum_{n=A}^B   \widetilde{m}_{\gamma,A}(n)\, a_n \,e_n.
\end{equation}
Letting $\widetilde{h}_{\gamma,A} = \sum_{n\in\mathbb Z} \widetilde{m}_{\gamma,A}(n)\,e_n$,  we have $\lVert f \rVert_{L^p} \leq A^{-\gamma} \, \lVert \widetilde{h}_{\gamma,A} \rVert_{L^1} \, \rVert g \rVert_{L^p}$ as in \eqref{Young}.
The same procedure as in \eqref{PoissonSummationFormula1} and \eqref{PoissonSummationFormula2} allows us to bound $\lVert  \widetilde{h}_{\gamma, A} \rVert_1 \leq \lVert (\widetilde{m}_\gamma)^\vee \rVert_{L^1(\mathbb R)}$,
which depends on $\gamma$ and $B/A$, 
so the right-hand side of \eqref{gammaneg} is established.
The same trick as in \eqref{EasyCaseInLemma} implies the left-hand side of \eqref{gammapos}.

In the case that $\gamma \in \mathbb N$, the dependence on $A/B$ can be dropped from the right-hand side of \eqref{gammapos}. 
For that, we work with $\chi\in \mathcal S(\mathbb R)$ such that
\begin{equation}\label{CutoffLemmaLuis3}
\chi(x) =  \begin{cases}
1 \quad \text{ if } \quad |x| \leq 1,\\
0 \quad \text{ if } \quad |x| \geq 2,
\end{cases}
\end{equation}
independent of $A/B$. The difference is that even if now $0 \in \operatorname{supp} \chi$, 
the multiplier $m_\gamma(\xi) = \xi^\gamma\chi(\xi)$ is Schwartz. 
Then, the proof works exactly the same as above, so that we get
\begin{equation}\label{IntermediateStepLemmaLuis2}
    \left\| f \right\|_p \leq B^\gamma \, \left\|  (m_\gamma)^\vee \right\|_{L^1(\mathbb R)} \left\| g \right\|_p,
\end{equation}
in such a way that now $\left\|  (m_\gamma)^\vee \right\|_{L^1(\mathbb R)}$ depends only on $\gamma$.  
For $\gamma \in \mathbb R$ such that $\gamma\geq 2$, we can also drop the dependence on $A/B$ 
from the right-hand side of \eqref{gammapos}.
We slightly change the multiplier $m_\gamma(\xi) = |\xi|^\gamma\chi(\xi)$, 
with $\chi$ as in \eqref{CutoffLemmaLuis3}.
However, we need to adapt the proof above 
because now $m_\gamma(\xi) \notin \mathcal S(\R)$, 
and thus not all steps are properly justified.

The procedure is the same up to \eqref{Young}.
In order to estimate $\lVert h_{\gamma, B} \rVert_1$, 
we need to check that we can use the Poisson summation formula in \eqref{PoissonSummationFormula1}.
According to \cite[Theorem 3.1.17]{Grafakos2009}, 
it is enough that the function $F_x(y) = m_{\gamma, B}(y)\, e^{2\pi i x y}$ 
is continuous and integrable, and that there exists $\delta >0$ such that
\begin{equation}\label{Poissoncond}
|F(y)|+|\widehat F(y)|\lesssim (1+|y|)^{-(1+\delta)}, \qquad \forall y \in \mathbb R.
\end{equation}
Since for all $x \in \mathbb T$ the function $F_x$ is continuous and has compact support, 
it is enough to bound $|\widehat F_x(y)| \lesssim (1+|y|)^{1+\delta}$.
We know that $\widehat F_x(y) = B \, (m_{\gamma})^{\vee}(B (x-y) ) = B \, \widehat{m_\gamma}(B(y-x))$, 
so it suffices to study the decay of $\widehat{m_\gamma}$.
Write 
\begin{equation}
\begin{split}
    \widehat{m_\gamma}(y) 
    & = \int_{\R} m_\gamma(\xi) \, e^{-2\pi i \xi y}\, \mathrm d\xi 
    = \int_{-2}^2  |\xi|^\gamma \, \chi(\xi) \, e^{-2\pi i \xi y} \, \mathrm d\xi  \\
    & = \int_0^2 \xi^\gamma \, \chi(\xi)\, e^{-2\pi i \xi y } \, \mathrm d\xi  
    +  \int_0^2 \xi^\gamma \, \chi(\xi)\, e^{2\pi i \xi y } \, \mathrm d\xi.
\end{split}
\end{equation}
The last two integrals are of the same kind, so we only bound the second one. 
First, we can directly bound
\begin{equation}\label{SmallY}
    \left| \int_0^2 \xi^\gamma \, \chi(\xi)\, e^{2\pi i \xi y } \, \mathrm d\xi \right| 
    \leq 2^\gamma, \qquad \forall |y| \leq 1/2.
\end{equation}
On the other hand, for $|y| > 1/2$
we split the integral
\begin{equation*}
\begin{split}
      \int_{0}^2  \xi^\gamma\, \chi(\xi) \, e^{2\pi i \xi y} d\xi 
    & \leq \int_{0}^{1/|y|}  \xi^\gamma \, \chi(\xi)\, e^{2\pi i \xi y} \, d\xi
     + \int_{1/|y|}^2  \xi^\gamma \, \chi(\xi) \, e^{2\pi i \xi y} \, d\xi \\
     & = I + II.
\end{split}
\end{equation*}
For the first term, we directly bound
\begin{equation}\label{Bound_For_I}
    |I| \leq \int_0^{1/|y|} \xi^\gamma \, d\xi = \frac{|y|^{-(1+\gamma)}}{1+\gamma}.
\end{equation}
For the second term, we have oscillations in the phase that we exploit integrating by parts.
Indeed, 
\begin{equation}\label{IntParts1}
\begin{split}
    II & = \int_{1/|y|}^2 \frac{\xi^\gamma \, \chi(\xi)}{2\pi i y} \, \partial_\xi (e^{2\pi i \xi y}) \, d\xi \\
    & = \frac{\xi^\gamma \, \chi(\xi)}{2\pi i y} \, e^{2\pi i \xi y} \Bigg|^2_{1/|y|}
    - \frac{1}{2\pi i y} \int_{1/|y|}^2 \partial_\xi( \xi^\gamma \, \chi(\xi) ) \, e^{2\pi i \xi y} \, d\xi \\
    & = \frac{\chi(1/|y|)}{2\pi i y \, |y|^\gamma} - \frac{1}{2\pi i y} \,  \int_{1/|y|}^2 \partial_\xi( \xi^\gamma \, \chi(\xi) ) \, e^{2\pi i \xi y} \, d\xi,
\end{split}
\end{equation}
and hence
\begin{equation}\label{Bound_For_II}
    |II| \lesssim \frac{1}{|y|^{1+\gamma}} + \frac{1}{|y|} |III|, \qquad \qquad III = \int_{1/|y|}^2 \partial_\xi( \xi^\gamma \, \chi(\xi) ) \, e^{2\pi i \xi y} \, d\xi.
\end{equation}
We integrate $III$ by parts again so that 
\begin{equation}
     III  = \frac{\partial_\xi( \xi^\gamma \, \chi(\xi) )}{2\pi i y}\, e^{2\pi i \xi y} \Bigg|^2_{1/|y|} - \frac{1}{2\pi i y} \int_{1/|y|}^2 \partial_\xi^2(\xi^\gamma\, \chi(\xi))\, e^{2\pi i \xi y }\, d\xi,
\end{equation}
and hence, expanding $\partial_\xi( \xi^\gamma \, \chi(\xi) ) = \gamma \xi^{\gamma - 1} \chi(\xi) + \xi^\gamma \chi'(\xi)$, we get
\begin{equation}
    |III| \lesssim \frac{1}{|y|^\gamma} + \frac{1}{|y|^{1+\gamma}} + \frac{1}{|y|} \int_0^2 |\partial_\xi^2(\xi^\gamma\, \chi(\xi))| \, d\xi.
\end{equation}
If $\gamma \geq 2$, then $|\partial_\xi^2(\xi^\gamma\, \chi(\xi))| \lesssim_\gamma 1$ in $\xi \in (0,2)$, so we get $|III| \lesssim_\gamma 1/|y|$ and therefore, from \eqref{Bound_For_II}, 
\begin{equation}
    |II| \lesssim_\gamma \frac{1}{|y|^{1+\gamma}} + \frac{1}{|y|^2} \lesssim \frac{1}{|y|^2}.
\end{equation}
Joining this with \eqref{Bound_For_I} we get
\begin{equation}\label{BigY}
    \left| \int_0^2 \xi^\gamma \, \chi(\xi)\, e^{2\pi i \xi y } \, dy \right| 
    \leq |I| + |II| 
    \lesssim_\gamma \frac{1}{|y|^2}, \qquad \forall |y| > 1/2.
\end{equation}
Thus, \eqref{SmallY} and \eqref{BigY} yield
\begin{equation}\label{BoundForMultiplier}
    |\widehat{m_\gamma}(y)| \lesssim_\gamma \frac{1}{(1+|y|)^2}, \qquad \forall y \in \mathbb R,
\end{equation}
which in turn implies
\begin{equation}
    |\widehat{F_x}(y) | = B \, | \widehat{m_\gamma}(B(y-x)) | \lesssim \frac{B}{(1 + B|x-y|)^2}
    \leq \frac{B}{(1 + |x-y|)^2}
    \leq  \frac{C_x \, B}{(1+|y|)^2}.
\end{equation}

Thus, we can use the Poisson summation formula as in \eqref{PoissonSummationFormula1}.
Then, by \eqref{PoissonSummationFormula2}, we obtain
\begin{equation}
    \lVert f \rVert_p \leq B^\gamma \, \lVert \widehat{m_\gamma} \rVert_{L^1(\R)} \, \lVert g \rVert_p,
\end{equation}
and since equation \eqref{BoundForMultiplier} implies $\lVert \widehat{m_\gamma} \rVert_{L^1(\R)} \lesssim_\gamma 1$, we get $\lVert f \rVert_p \lesssim_\gamma B^\gamma \, \lVert g \rVert_p,$ and conclude. 
\end{proof}

\begin{cor}\label{CorollaryLuis}
Let $0 < A < B$ possibly vary according to some parameter, and assume that the ratio $B/A$ is bounded by a constant that does not depend on that parameter.
Then, for all $p \geq 1$ and $\gamma \in \mathbb{R}$, 
\begin{equation}\label{lemmaluisforHPfilters}
 \left\| \sum_{n=A}^{B} n^\gamma \, a_n\, e_n \right\|_{p} \simeq_\gamma
  A^\gamma \, \left\| \sum_{n=A}^{B}a_n e_n \right\|_{p} \simeq_\gamma
   B^\gamma \, \left\| \sum_{n=A}^{B}a_n e_n \right\|_{p}.
\end{equation}
\end{cor}
\begin{proof}
If $1 < B/A \leq C$ with $C>0$ a universal constant, the dependence in $B/A$ can be dropped from the proof of Lemma~\ref{lemmaluis} by choosing the plateau region of $\widetilde{\chi}$ to be $[1,C]$ instead of $[1,B/A]$. Also, $A \simeq B$ holds independently of the parameters.
\end{proof}

The following lemma will be necessary to work with the structure functions of $R_s$ in Section~\ref{SECTION_StructureFunctions}. 
\begin{lem}\label{Application2}
Let $p \geq 1$, $0 < \ell < 1/2$ and $B \geq 1$. If $s\leq 1$,
\begin{equation}\label{lemmaluisforstrucfunc}
    \Bigg\| \sum_{n=1}^{\sqrt{B}} \frac{\sin(\pi n^2 \ell)}{n^{2s}}\, e_{n^2} \Bigg\|_p
    \lesssim_s
     \ell \, \Bigg\lVert \sum_{n=1}^{\sqrt{B}}  n^{2(1-s)}\, e_{n^2} \Bigg\rVert_p
     +  B^{3-s}\, \ell^3 \, \cosh{(2\pi B \ell)} \,  \Bigg\| \sum_{n=1}^{\sqrt{B}}  e_{n^2} \Bigg\|_p.
\end{equation}
If $1 < s \leq 3$, then 
\begin{equation}
\begin{split}
    \Bigg\| \sum_{n=1}^{\sqrt{B}} \frac{\sin(\pi n^2 \ell)}{n^{2s}} \, e_{n^2} \Bigg\|_p 
    & \lesssim 
    \ell \, \Bigg\| \sum_{n=1}^{\sqrt{B}} \,  \frac{e_{n^2}}{n^{2(s-1)}} \Bigg\|_p 
    + 
    \ell^3 \, \Bigg\| \sum_{n=1}^{\sqrt{B}} \, n^{2(3 - s)} \, e_{n^2} \Bigg\|_p   \\
    & \qquad \qquad \qquad \qquad \qquad \qquad  +
     B^{5-s}\, \ell^5\,  \cosh(2\pi \,  \ell B) \, \Bigg\| \sum_{n=1}^{\sqrt{B}} e_{n^2} \Bigg\|_p.
\end{split}
\end{equation}
\end{lem}

To prove it, we need one more auxiliary lemma.
\begin{lem}\label{lemmaVa}
Let $C>0$ and $M\geq1$.
Let also $f\in H^1(\R)$ with support in $[-M,M]$.
Then
\begin{equation*}
\sum_{k = 0}^{\infty}\frac{C^{2k}\|\partial^{2k}\hat f\|_{L^1(\R)}}{(2k+1)!}\lesssim \cosh(2\pi CM) \|f\|_{H^1(\R)}.
\end{equation*}
\end{lem}

\begin{proof}[Proof of Lemma~\ref{lemmaVa}]
Let us first prove 
\begin{equation}\label{interpoL1}
 \|g\|_{L^1(\R)} \leq 2\sqrt{2}\, \|g\|_{L^2(\R)}^{1/2}\|xg\|_{L^2(\R)}^{1/2}, \qquad \forall g \in L^2(\mathbb{R}) \, \text{ such that } x g \in L^2(\mathbb R). 
\end{equation}
For that, let $R>0$ and $g\in L^2(\mathbb{R})$ a nonzero function. The Cauchy-Schwarz inequality implies that
\begin{equation}
 \lVert g \rVert_{L^1(\R)}=\int_{|x| < R} |g(x)| \,\mathrm{d}x+\int_{|x|>R} \frac{ |x g(x)|}{|x|}\,\mathrm{d}x \leq \sqrt{2R}\, \|g\|_{L^2(\R)}+\sqrt{2/R}\,\|xg\|_{L^2(\R)},
\end{equation}
and choosing $R=\|xg\|_{L^2(\R)}/\|g\|_{L^2(\R)}$ \eqref{interpoL1} follows. 

Let now $M \geq 1$ and $f\in H^1(\R)$ a nonzero function with support in $[-M,M]$. Let us prove
\begin{equation}\label{boundd2k}
 \|\partial^{2k}\widehat f\|_{L^1(\R)} \lesssim  \sqrt{2k+1} \, (2\pi M)^{2k} \, \|f\|_{H^1(\R)}, \qquad \forall k \in \mathbb{N}\cup \{0\}.
\end{equation}
The case $k=0$ immediately follows using \eqref{interpoL1} for $\widehat{f}$ because 
\begin{equation}
 \| \widehat f \|_{L^1(\R)} \lesssim \|\widehat f\|_{L^2(\R)}^{1/2}\|\xi\widehat f\|_{L^2(\R)}^{1/2} \lesssim  \|f\|_{H^1(\R)}.
\end{equation}
For $k>0$, since $\partial^{2k}\widehat{f} = (-2\pi i)^{2k}\,  (x^{2k}\,f)^{\wedge}$ and $2\pi i\,\xi\widehat{f} = (\partial f)^\wedge$, using \eqref{interpoL1} with $(x^{2k}f)^\wedge$ and the Plancherel theorem we have
\begin{align}
 \lVert \partial^{2k}\widehat f \rVert_{L^1(\R)} & \lesssim (2\pi)^{2k}  \lVert x^{2k}f \rVert_{L^2(\R)}^{1/2}\|\partial_x(x^{2k}f)\|_{L^2(\R)}^{1/2} \\
 & \lesssim (2\pi)^{2k} \|x^{2k}f\|_{L^2(\R)}^{1/2} \left( 2k\|x^{2k-1}f\|_{L^2(\R)} + \|x^{2k}f'\|_{L^2(\R)}  \right)^{1/2}.
\end{align}
Since $f$ is supported in $[-M,M]$ and $M\geq1$, we get
\begin{align}
 \|\partial^{2k}\widehat f\|_{L^1(\R)} & \lesssim  (2\pi M)^{2k} \|f\|_{L^2(\R)}^{1/2}\left(1+2k\right)^{1/2}\|f\|_{H^1(\R)}^{1/2} \lesssim  \sqrt{2k+1} \, (2\pi M)^{2k} \, \|f\|_{H^1(\R)} ,
\end{align}
which is \eqref{boundd2k}. Thus, for $C>0$, \eqref{boundd2k} implies that
\begin{equation}
\sum_{k=0}^{\infty}\frac{C^{2k}\|\partial^{2k}\widehat f\|_{L^1(\R)}}{(2k+1)!} \lesssim  \sum_{k=0}^{\infty}\frac{C^{2k}(2\pi M)^{2k}\sqrt{2k+1}}{(2k+1)!} \|f\|_{H^1(\R)}  \lesssim   \cosh(2\pi CM) \|f\|_{H^1(\R)}.
\end{equation}
\end{proof}

With Lemmas~\ref{lemmaluis} and \ref{lemmaVa}, we are ready to prove Lemma~\ref{Application2}.
\begin{proof}[Proof of Lemma~\ref{Application2}]
By the Taylor series expansion of the sine and the triangle inequality, we obtain 
\begin{equation}\label{AuxCor1}
    \begin{split}
        \left\| \sum_{n=1}^{\sqrt{B}} \frac{\sin(\pi n^2 \ell)}{n^{2s}} e_{n^2} \right\|_p
        & =  \left\| \sum_{n=1}^{\sqrt{B}} \left( \sum_{k=0}^{\infty} (-1)^k \, \frac{(\pi n^2 \ell)^{2k+1}}{(2k+1)!}\right) \, \frac{e_{n^2}}{n^{2s}} \right\|_p \\
        & \leq \sum_{k=0}^{\infty}\frac{(\pi\ell)^{2k+1}}{(2k+1)!} \left\| \sum_{n=1}^{\sqrt{B}}  n^{2(2k + 1 - s)}\, e_{n^2} \right\|_p \\
        & 
        = \pi \ell \, \left\lVert \sum_{n=1}^{B}n^{1-s}\, \sigma_n\, e_{n} \right\rVert_p
        + \sum_{k=0}^{\infty}\frac{(\pi\ell)^{2k+3}}{(2k+3)!} \left\| \sum_{n=1}^{B}  n^{2k + 3 - s}\, \sigma_n\,  e_{n} \right\|_p,
    \end{split}
\end{equation}
where $\sigma_n$ is the characteristic function of the square integers.
If $s \leq 1$, then $ 3-s \geq 2 $, 
so by Lemma~\ref{lemmaluis} we take $n^{2(3-s)}$ out of the second $L^p$ norm to get
\begin{equation}
    \left\| \sum_{n=1}^{\sqrt{B}} \frac{\sin(\pi n^2 \ell)}{n^{2s}} \, e_{n^2} \right\|_p \lesssim_s 
    \ell \, \left\lVert \sum_{n=1}^{B}  n^{1-s}\, \sigma_n \, e_{n} \right\rVert_p
    +
    B^{3-s}\, \sum_{k=0}^{\infty}\frac{(\pi\ell)^{2k+3}}{(2k+3)!}\, \left\| \sum_{n=1}^{B}  n^{2k}\, \sigma_n \, e_{n} \right\|_p.
\end{equation}
We bring $n^{2k}$ outside too, but now we need to keep control of the constants that depend on $k$. 
For that, we see in the proof of Lemma~\ref{lemmaluis}, in \eqref{IntermediateStepLemmaLuis2}, 
that we get the bound 
\begin{equation}
 \left\| \sum_{n=1}^B  n^{2k}\sigma_n e_n \right\|_p \leq B^{2k} \left\| (m_{2k})^\vee \right\|_{L^1(\R)}\left\| \sum_{n=1}^B \sigma_n e_n \right\|_p, 
\end{equation}
where $m_{2k}(\xi) = \xi^{2k} \, \chi(\xi)$ and $\chi \in \mathcal S(\mathbb R)$ defined in \eqref{CutoffLemmaLuis3}.  Write $ (m_{2k})^\vee = (2\pi i )^{-2k}\partial^{2k} \check\chi$ so that
\begin{equation}
\begin{split}
    \left\| \sum_{n=1}^{\sqrt{B}} \frac{\sin(\pi n^2 \ell)}{n^{2s}}\,  e_{n^2} \right\|_p 
    &  \lesssim_s
    \ell \, \left\lVert \sum_{n=1}^{B}  n^{1-s}\, \sigma_n \, e_{n} \right\rVert_p \\
    & \qquad \qquad  +
    B^{3-s}\, \ell^3 \,  \Bigg\| \sum_{n=1}^B \sigma_n e_n \Bigg\|_p \,  \sum_{k=0}^{\infty}  \left(\frac{B \ell}{2}\right)^{2k} \frac{\| \partial^{2k}\check\chi \|_{L^1(\R)}}{(2k+3)!}. 
\end{split}
 \end{equation}
Applying Lemma~\ref{lemmaVa} to $\chi_-(x) = \chi(-x)$ with $M=2$, we obtain
\begin{equation}
    \left\| \sum_{n=1}^{\sqrt{B}} \frac{\sin(\pi n^2 \ell)}{n^{2s}}\, e_{n^2} \right\|_p
    \lesssim_s
     \ell \, \Bigg\lVert \sum_{n=1}^{\sqrt{B}}  n^{2(1-s)}\, e_{n^2} \Bigg\rVert_p
     + \left\| \chi_- \right\|_{H^1(\R)} \, B^{3-s}\, \ell^3 \, \cosh{(2\pi B \ell)} \,  \Bigg\| \sum_{n=1}^{\sqrt{B}}  e_{n^2} \Bigg\|_p.
\end{equation}

If $1 < s \leq 3$, 
we take one more term from the sum in \eqref{AuxCor1} and write
\begin{equation}
\begin{split}
    \Bigg\| \sum_{n=1}^{\sqrt{B}} \frac{\sin(\pi n^2 \ell)}{n^{2s}} \, e_{n^2} \Bigg\|_p
    & 
    \lesssim 
    \ell \, \Bigg\| \sum_{n=1}^{\sqrt{B}} \,  \frac{e_{n^2}}{n^{2(s-1)}} \Bigg\|_p 
    + 
    \ell^3 \, \Bigg\| \sum_{n=1}^{\sqrt{B}} \, n^{2(3 - s)} \, e_{n^2} \Bigg\|_p  
    \\
    & \qquad \qquad 
    +
    \sum_{k=0}^\infty \frac{(\pi \ell)^{2k+5}}{(2k+5)!} \, \Bigg\lVert  \sum_{n=1}^B n^{2k + 5 - s} \, \sigma_n \, e_n  \Bigg\rVert_p.
\end{split}
\end{equation}
For the last term,
since now $2k+5-s \geq 2$, we can do like when $s\leq 1$ and get
\begin{equation}
    \begin{split}
    \sum_{k=0}^{\infty}\frac{(\pi\ell)^{2k+5}}{(2k+5)!} \Bigg\| \sum_{n=1}^B  n^{2k + 5-s}\sigma_n e_n \Bigg\|_p 
    & 
    \leq B^{5-s}\, \ell^5 \,  \Bigg\| \sum_{n=1}^{\sqrt{B}} e_{n^2} \Bigg\|_p  \, \sum_{k=0}^{\infty} \left( \frac{B \ell}{2} \right)^{2k} \frac{ \lVert \partial^{2k} \check{\chi} \rVert_{L^1(\mathbb R)} }{(2k+5)!} 
    \\
    &
    \lesssim \left\| \chi_- \right\|_{H^1(\R)}\, B^{5-s}\, \ell^5\,  \cosh(2\pi \,  \ell B) \, \left\| \sum_{n=1}^{\sqrt{B}} e_{n^2} \right\|_p,
    \end{split}
\end{equation}
and the lemma is proved.
\end{proof}

\begin{rmk}
We do not write the expressions for the cases $s>3$ in Lemma~\ref{Application2},
but they can be obtained iterating the proof.
\end{rmk}

\section{Estimates on \texorpdfstring{$L^p$}{lp} norms}\label{SECTION_LpNorms}

We begin by recalling an important estimate on the $L^p$ norms of the functions 
\begin{equation}\label{ZalcwasserFunction}
    K_N(x) = \sum_{n=1}^N e^{2\pi i n^2 x}, \qquad N \in \mathbb N, \quad N \gg 1,
\end{equation}
due to Zalcwasser \cite[p.28]{Zalcwasser}.
\begin{thm}[Zalcwasser]\label{TheoremZalcwasser}
Let $p > 0$ and $N \in \mathbb N$. Then, 
\begin{equation}\label{Zalcwasser}
    \left\lVert K_N \right\rVert_p \simeq_p \psi_p(N) = \left\{ \begin{array}{ll}
        N^{1/2}, & p < 4, \\
        N^{1/2}\, \left( \log N \right)^{1/4},  & p=4, \\
        N^{1-2/p}, & p>4. 
    \end{array}
    \right.
\end{equation}
\end{thm}

Thinking of high-pass filters, from Zalcwasser's theorem~\ref{TheoremZalcwasser} and Lemma~\ref{lemmaluis} we deduce the following behavior for Littlewood-Paley blocks of $R_s$. 

\begin{prop}\label{PropDeltak}
Let $p > 0$, $s \in \mathbb R$ and $A, k \in \mathbb N$. Then, 
\begin{equation}\label{EstimationLPPiece}
\left\lVert  \sum_{n=A^k}^{A^{k+1}-1} \frac{e^{2\pi i n^2 x}}{n^{2s}} \right\rVert_p \simeq_p \frac{\lVert K_{A^k} \rVert_p}{A^{2ks}} \simeq_p \left\{ \begin{array}{ll}
    A^{k(1/2 - 2s)}, & p < 4,  \\
    A^{k(1/2 - 2s)}\, k^{1/4},  & p=4, \\
    A^{k(1 - 2/p - 2s)}, & p > 4.
\end{array}
\right.
\end{equation}
\end{prop}

\begin{proof}
Assume first that $p\geq 1$. Since $(A^{k+1}-1 )/A^k \leq A$, which is independent of $k$, by Corollary~\ref{CorollaryLuis} we write
\begin{equation}\label{BeforeTriangleInequality}
    \left\lVert \sum_{n=A^k}^{A^{k+1}-1} \frac{e^{2\pi i n^2 x}}{n^{2s}} \right\rVert_p \simeq \frac{1}{A^{2ks}} \, \left\lVert \sum_{n=A^k}^{A^{k+1}-1} e^{2\pi i n^2 x} \right\rVert_p = \frac{1}{A^{2ks}} \, \left\lVert K_{A^{k+1}-1} - K_{A^k-1} \right\rVert_p 
\end{equation}
By the triangle inequality and Theorem~\ref{TheoremZalcwasser},  \eqref{BeforeTriangleInequality} is bounded from above by 
\begin{equation}\label{UpperBound}
    \frac{1}{A^{2ks}} \, \left( \left\lVert K_{A^{k+1}-1}  \right\rVert_p + \left\lVert K_{A^k-1} \right\rVert_p  \right) \lesssim  \frac{1}{A^{2ks}} \,  \left\lVert K_{A^k} \right\rVert_p, 
\end{equation}
and from below (see Remark~\ref{RemarkLittlewoodPaley}) by 
\begin{equation}\label{LowerBound}
    \frac{1}{A^{2ks}} \, \left| \left\lVert K_{A^{k+1}-1}  \right\rVert_p - \left\lVert K_{A^k-1} \right\rVert_p  \right| \gtrsim \frac{1}{A^{2ks}} \, \left\lVert K_{A^k} \right\rVert_p, 
\end{equation}
so the statement follows.

Once the result holds for $p \geq 1$, let us prove it for $0<p<1$. We do it by interpolation using H\"older's inequality as follows. For simplicity, call $\Delta_k(x) = \sum_{n=A^k}^{A^{k+1}-1} e^{2\pi i n^2 x}/n^{2s}$. Then, 
\begin{equation}\label{ExtrapolationByHolder}
    \left\lVert \Delta_k \right\rVert_p^p =  \int_{\mathbb T} \left| \Delta_k (x) \right|^p \, dx \, \, \,  \lesssim  \, \, \left( \int_{\mathbb T} \left| \Delta_k (x) \right|^{p(2/p)}\, dx \right)^{p/2} = \lVert \Delta_k \rVert_2^p, \qquad \forall p \in (0,2).
\end{equation}
We already computed $\lVert \Delta_k \rVert_2$ above, so we get $ \left\lVert \Delta_k \right\rVert_p \lesssim A^{k(1/2-2s)}$. For the lower bound, we interpolate 2 between $p$ and 3. Indeed, if $p<2<3$, 
\begin{equation}\label{Interpolation1}
    \left\lVert \Delta_k \right\rVert_2^2 = \int_{\mathbb T} \left| \Delta_k \right|^2 \leq \left( \int_{\mathbb T} \left| \Delta_k \right|^{a\theta} \right)^{1/\theta} \, \left( \int_{\mathbb T} \left| \Delta_k \right|^{b\theta'} \right)^{1/\theta'} = \lVert \Delta_k \rVert_p^{p/\theta}\,\lVert \Delta_k \rVert_3^{3/\theta'},
\end{equation}
where $a+b=2$, $a\theta=p$, $b\theta' = 3$, $1/\theta + 1/\theta'=1$ and $\theta >1$. This implies $\theta = 3-p>1$, so by the result for $p\geq 1$ we get
\begin{equation}\label{Interpolation2}
    \lVert \Delta_k \rVert_p \geq \frac{\left\lVert \Delta_k \right\rVert_2^{2\theta/p}}{\left\lVert \Delta_k \right\rVert_3^{3\theta/(p\theta')}} \simeq \left(A^{k(1/2-2s)} \right)^{(2\theta - 3\theta/\theta')/p} = A^{k(1/2-2s)}.
\end{equation}
\end{proof}
\begin{rmk}\label{RemarkLittlewoodPaley}
To be precise, for \eqref{LowerBound} to hold we need to take the Littlewood-Paley parameter $A$ large enough. Indeed, if $0 < c_p < C_p$ are the underlying constants in Theorem~\ref{TheoremZalcwasser}, then we have
\begin{equation}
\begin{split}
    \left\lVert K_{A^{k+1}-1}  \right\rVert_p - \left\lVert K_{A^k-1} \right\rVert_p & \geq c_p \psi_p(A^{k+1})/2 - C_p \psi_p(A^k)  = \psi_p(A^k) \left( \frac{c_p}{2}  \frac{\psi_p(A^{k+1})}{\psi_p(A^k)} - C_p \right) \\
    & \geq \psi_p(A^k) \left( \frac{c_p}{2} A^{1/2} - C_p \right),
\end{split}
\end{equation}
so we need to have $A = A_p$ large enough such that $c_p A^{1/2}/2 - C_p \geq 1$. 
\end{rmk}

In the setting of the structure functions, we will need to compute the behavior of the low-pass filters of the functions $\sum_{n=1}^\infty n^{-2s}e^{2\pi i n^2 x}$. The following theorem can be seen as a generalization of Zalcwasser's Theorem~\ref{TheoremZalcwasser}.

\begin{thm}\label{TheoremNewZalcwasser}
Let $s \in \mathbb R$, $p > 0$ and $N \in \mathbb N$. Then,
\begin{itemize}
    \item If $s < 1/4$, 
    \begin{equation}
        \left\lVert \sum_{n=1}^N \frac{e^{2\pi i n^2 x}}{n^{2s}}  \right\rVert_p \simeq_{s,p} \left\{  \begin{array}{ll}
            N^{1/2 - 2s}, & p < 4,  \\
            N^{1/2 - 2s}\, (\log N )^{1/4}, & p = 4, \\
            N^{1 - 2/p - 2s}, & p>4,
        \end{array}
        \right.
    \end{equation}
    
    \item If $s=1/4$, 
    \begin{equation}
        \left\lVert \sum_{n=1}^N \frac{e^{2\pi i n^2 x}}{n^{1/2}}  \right\rVert_p \simeq_{p} \left\{  \begin{array}{ll}
            (\log N)^{1/2}, & p < 4,  \\
            N^{1/2 - 2/p}, & p > 4,
        \end{array}
        \right.
    \end{equation}
    and when $p=4$,
    \begin{equation}
        (\log N)^{1/2} \lesssim  \left\lVert \sum_{n=1}^N \frac{e^{2\pi i n^2 x}}{n^{1/2}}  \right\rVert_4  \lesssim (\log N)^{3/4}. 
    \end{equation}
    
    \item If $1/4 < s < 1/2$, 
    \begin{equation}
        \left\lVert \sum_{n=1}^N \frac{e^{2\pi i n^2 x}}{n^{2s}}  \right\rVert_p \simeq_{s,p} \left\{  \begin{array}{ll}
            1, & p < \frac{2}{1-2s},  \\
            N^{1 - 2/p - 2s}, & p > \frac{2}{1-2s},
        \end{array}
        \right.
    \end{equation}
    and when $p=2/(1-2s)$,
    \begin{equation}
        (\log N)^{1/p} \lesssim_s  \left\lVert \sum_{n=1}^N \frac{e^{2\pi i n^2 x}}{n^{2s}}  \right\rVert_{p}  \lesssim_s (\log N)^{1/2}, \qquad  p=\frac{2}{1-2s}.
    \end{equation}

    \item If $s \geq 1/2$, 
    \begin{equation}
        \left\lVert \sum_{n=1}^N \frac{e^{2\pi i n^2 x}}{n^{2s}}  \right\rVert_p \simeq_{s,p} 1, \qquad \forall p > 0.
    \end{equation}
\end{itemize}
\end{thm}
\begin{conj}
It seems reasonable to conjecture that 
\begin{equation}
     \left\lVert \sum_{n=1}^N \frac{e^{2\pi i n^2 x}}{n^{1/2}}  \right\rVert_4  \simeq (\log N)^{3/4}. 
\end{equation}
\end{conj}

\begin{proof}[Proof of Theorem~\ref{TheoremNewZalcwasser}]
We prove the theorem classifying in $p$ instead of in $s$. The proof is based in the Littlewood-Paley theorem, Minkowski's integral inequality, Zalcwasser's Theorem~\ref{TheoremZalcwasser} and Proposition~\ref{PropDeltak}. Let $N \in \mathbb N$. 
Decompose the function as
\begin{equation}\label{LittlewoodPaleyDecompositionInProof}
    \sum_{n=1}^N \frac{e^{2\pi i n^2 x}}{n^{2s}} 
    = \sum_{k=0}^{k(N)-1} \Delta_k(x) + \widetilde\Delta_{k(N)}(x), 
\end{equation}
where we define $k(N)$ as the only integer that satisfies $ A^{k(N)} \leq N < A^{k(N)+1} $ and $\Delta_k(x)$ are the Littlewood-Paley blocks 
\begin{equation}
    \Delta_k(x) = \sum_{n=A^k}^{A^{k+1}-1} \frac{e^{2\pi i n^2 x}}{n^{2s}} \qquad \forall 1 \leq k < k(N), \qquad \quad 
    \widetilde\Delta_{k(N)}(x) = \sum_{n=A^{k(N)}}^{N} \frac{e^{2\pi i n^2 x}}{n^{2s}}.
\end{equation}
Here we are taking $A \in \mathbb N$ large enough according to Remark~\ref{RemarkLittlewoodPaley}.
In particular, we have $ k(N) \simeq \log N$, since by definition we have $ \frac{\log N}{\log A} - 1 < k(N) \leq \frac{\log N}{\log A}$.
Thus, by the Littlewood-Paley theorem, for $1 < p < \infty$ we have
\begin{equation}\label{LittlewoodPaleyInProof}
 \left\lVert \sum_{n=1}^N \frac{e^{2\pi i n^2 x}}{n^{2s}} \right\rVert_p \simeq_p \left\lVert \left(  \sum_{k=0}^{k(N) - 1} |\Delta_k(x)|^2 + |\widetilde\Delta_{k(N)}(x)|^2  \right)^{1/2} \right\rVert_p.
\end{equation}
If $p\geq 2$, Minkowski's integral inequality allows us to bound \eqref{LittlewoodPaleyInProof} from above in terms of $\lVert \Delta_k \rVert_p$. Thus, we get 
\begin{equation}\label{UpperBoundAfterMinkowski}
   \left\lVert \sum_{n=1}^N \frac{e^{2\pi i n^2 x}}{n^{2s}} \right\rVert_p
   \lesssim \left( \sum_{k=0}^{k(N)-1} \lVert \Delta_k \rVert_p^2 + \lVert \widetilde\Delta_{k(N)} \rVert_p^2  \right)^{1/2} 
   \lesssim_p \left( \sum_{k=0}^{k(N)}  \frac{\lVert K_{A^k} \rVert_p^2 }{A^{4ks}} \right)^{1/2}.
\end{equation}
To get the last equality we used Proposition~\ref{PropDeltak}. 
Indeed, for $k < k(N)$ we have $\lVert \Delta_k \rVert_p \simeq_p \lVert K_{A_k} \rVert_p / A^{2ks}$.
For $k = k(N)$, we can bound
\begin{equation}
    \lVert \widetilde\Delta_{k(N)} \rVert_p 
    = \lVert K_N - K_{A^{k(N)}-1} \rVert_p 
    \leq \lVert K_N \rVert_p + \rVert K_{A^{k(N)}-1} \rVert_p 
    \simeq \rVert K_{A^{k(N)}} \rVert_p
\end{equation}
because $N \simeq A^{k(N)}$. 
Now, we estimate the norm $\lVert K_{A^k} \rVert_p$ using Theorem~\ref{TheoremZalcwasser}, for which we need to split cases.
\begin{itemize}
    \item $2\leq p<4$: In this case,  $\lVert K_{A^k} \rVert_p^2 \simeq A^{k} $, so recalling that $k(N) \simeq \log N$, we bound \eqref{UpperBoundAfterMinkowski} from above to deduce
    \begin{equation}\label{UpperBoundByLP}
        \left\lVert \sum_{n=1}^N \frac{e^{2\pi i n^2 x}}{n^{2s}} \right\rVert_p \lesssim_p \left(  \sum_{k=0}^{k(N)}   \frac{1}{A^{k(4s-1)}}  \right)^{1/2} \lesssim \left\{ \begin{array}{ll}
            N^{1/2 - 2s}, &  s < 1/4, \\
            (\log N)^{1/2}, &  s=1/4, \\
            1, &  s > 1/4.
        \end{array}
        \right.
    \end{equation}
    This is optimal because H\"older's inequality and Plancherel's theorem suffice to check
    \begin{equation}\label{LowerBoundByL2}
        \left\lVert \sum_{n=1}^N \frac{e^{2\pi i n^2 x}}{n^{2s}} \right\rVert_p \gtrsim \left\lVert \sum_{n=1}^N \frac{e^{2\pi i n^2 x}}{n^{2s}} \right\rVert_2 = \left( \sum_{n=1}^N \frac{1}{n^{4s}} \right)^{1/2} \simeq \left\{ \begin{array}{ll}
            N^{1/2 - 2s}, &  s < 1/4, \\
            (\log N)^{1/2}, &  s=1/4, \\
            1, &  s > 1/4.
        \end{array}
        \right.
    \end{equation}
    
    \item  $p=4$: Theorem~\ref{TheoremZalcwasser} gives $\lVert K_{A^k} \rVert_4^2 \simeq_p A^k \left( \log A^k \right)^{1/2} \simeq k^{1/2}\, A^k $ in this case, so bounding \eqref{UpperBoundAfterMinkowski} we get
    \begin{equation}\label{UpperBoundNewZalcwasserp4}
        \left\lVert \sum_{n=1}^N \frac{e^{2\pi i n^2 x}}{n^{2s}} \right\rVert_4 \lesssim \left(  \sum_{k=0}^{k(N)}   \frac{k^{1/2}}{A^{k(4s-1)}}   \right)^{1/2} \lesssim \left\{ \begin{array}{ll}
            N^{1/2 - 2s} (\log N)^{1/4}, &  s < 1/4, \\
            (\log N)^{3/4}, &  s=1/4, \\
            1, &  s > 1/4.
        \end{array}
        \right.
    \end{equation}
    For the lower bound, the inclusion $\ell^2 \subset \ell^\infty$ together with the Littlewood-Paley decomposition in \eqref{LittlewoodPaleyInProof} and Proposition~\ref{PropDeltak} allow us to bound 
    \begin{equation}\label{LowerBoundInNewZalcwasser}
        \left\lVert \sum_{n=1}^N \frac{e^{2\pi i n^2 x}}{n^{2s}} \right\rVert_4 \gtrsim \left\lVert \Delta_{k(N)} \right\rVert_4 \simeq N^{1/2-2s} (\log N )^{1/4},
    \end{equation}
    which according to \eqref{UpperBoundNewZalcwasserp4} is optimal when $s<1/4$. For $s=1/4$, a better but still non-optimal lower bound $(\log N)^{1/2}$ is obtained by the $L^2$ estimate in \eqref{LowerBoundByL2}. 
    For $s>1/4$, the $L^2$ estimate in \eqref{LowerBoundByL2} gives the optimal lower bound 1.

    \item $p > 4$: Now $\lVert K_{A^k} \rVert_p^2 \simeq_p A^{k(2-4/p)} $ from Theorem~\ref{TheoremZalcwasser}, so bound \eqref{UpperBoundAfterMinkowski} with that to get
    \begin{equation}\label{UpperBoundNewZalcwasserpGreaterThan4}
        \left\lVert \sum_{n=1}^N \frac{e^{2\pi i n^2 x}}{n^{2s}} \right\rVert_p \lesssim_p \left(  \sum_{k=0}^{k(N)}   A^{k(2-4/p - 4s)}   \right)^{1/2} \lesssim \left\{ \begin{array}{ll}
            N^{1 - 2/p - 2s}, &  s < 1/2 - 1/p, \\
            (\log N)^{1/2}, &  s=1/2 - 1/p , \\
            1, &  s > 1/2 - 1/p.
        \end{array}
        \right.
    \end{equation}
    To get the lower bounds, as in \eqref{LowerBoundInNewZalcwasser}, from \eqref{LittlewoodPaleyInProof} and Proposition~\ref{PropDeltak} we get
    \begin{equation}\label{LowerBoundInNewZalcwasser2}
        \left\lVert \sum_{n=1}^N \frac{e^{2\pi i n^2 x}}{n^{2s}} \right\rVert_p \gtrsim_p \left\lVert \Delta_{k(N)} \right\rVert_p \simeq N^{1 - 2/p - 2s},
    \end{equation}
    which in view of \eqref{UpperBoundNewZalcwasserpGreaterThan4} is optimal for $s < 1/2 - 1/p$. 
    For $s > 1/2 - 1/p$, the $L^2$ estimate in \eqref{LowerBoundByL2} gives the optimal lower bound 1. 
    For the critical $s=1/2-1/p$, both \eqref{LowerBoundInNewZalcwasser2} 
    and the $L^2$ estimate give the lower bound 1. A better estimate can be obtained if in the Littlewood-Paley decomposition \eqref{LittlewoodPaleyInProof} we use the inclusion $\ell^2 \subset \ell^p$. Indeed, 
    \begin{equation}
    \begin{split}
     \left\lVert \sum_{n=1}^N \frac{e^{2\pi i n^2 x}}{n^{2s}}  \right\rVert_p  
     & \simeq_p \left\lVert \left( \sum_{k=0}^{k(N)-1} |\Delta_k(x)|^2 + |\widetilde\Delta_{k(N)}(x)|^2  \right)^{1/2} \right\rVert_p \\
     & \geq \left\lVert \left( \sum_{k=0}^{k(N)-1} |\Delta_k(x)|^p \right)^{1/p} \right\rVert_p \\
     & = \left( \sum_{k=0}^{k(N)-1} \lVert \Delta_k \rVert_p^p  \right)^{1/p}. 
     \end{split}
    \end{equation}
    By Proposition~\ref{PropDeltak} with $p>4$ and $s=1/2-1/p$, we have $\lVert \Delta_k \rVert_p \simeq 1$, so the above is bounded below by $(\log N)^{1/p}$. In view of \eqref{UpperBoundNewZalcwasserpGreaterThan4}, though, this still leaves the behavior of the critical case open.
    
\end{itemize}

We are missing the case $0<p<2$. Like in \eqref{ExtrapolationByHolder}, H\"older's inequality directly gives
\begin{equation}
     \left\lVert \sum_{n=1}^N \frac{e^{2\pi i n^2 x}}{n^{2s}}  \right\rVert_p \, \, \,  \lesssim  \, \, \,  \left\lVert \sum_{n=1}^N \frac{e^{2\pi i n^2 x}}{n^{2s}}  \right\rVert_2.
\end{equation}
For the lower bound, we do like in \eqref{Interpolation1} and \eqref{Interpolation2} with $\theta = 3-p > 1$ and get 
\begin{equation}
    \left\lVert \sum_{n=1}^N \frac{e^{2\pi i n^2 x}}{n^{2s}}  \right\rVert_p \geq \frac{\left\lVert \sum_{n=1}^N \frac{e^{2\pi i n^2 x}}{n^{2s}}  \right\rVert_2^{2\theta/p}}{\left\lVert \sum_{n=1}^N \frac{e^{2\pi i n^2 x}}{n^{2s}}  \right\rVert_3^{3\theta/(p\theta')}} \simeq \left\lVert \sum_{n=1}^N \frac{e^{2\pi i n^2 x}}{n^{2s}}  \right\rVert_2
\end{equation}
because $2\theta/p-3\theta/(p\theta')=1$, and we saw in \eqref{UpperBoundByLP} and \eqref{LowerBoundByL2} that the $L^p$ norms for all $2 \leq p < 4$ are all equivalent. Thus, \eqref{UpperBoundByLP} and \eqref{LowerBoundByL2} are also valid for $0<p<2$.

Summing up, we have proved  
 \begin{equation}
     \left\lVert \sum_{n=1}^N \frac{ e^{2\pi i n^2 x}}{n^{2s}} \right\rVert_p \simeq_p \left\{ \begin{array}{ll}
         N^{1/2 - 2s}, & s < 1/4, \\
         (\log N)^{1/2}, & s=1/4, \\
         1, & s > 1/4,
     \end{array}
     \right.
     \qquad \text{ for } 0 < p < 4,
 \end{equation}
 
 \begin{equation}
     \left\lVert \sum_{n=1}^N \frac{ e^{2\pi i n^2 x}}{n^{2s}} \right\rVert_4 \simeq \left\{ \begin{array}{ll}
         N^{1/2 - 2s} \, (\log N )^{1/4}, & s < 1/4, \\
         1, & s > 1/4,
     \end{array}
     \right.
 \end{equation}
 
 \begin{equation}
     (\log N )^{1/2} \lesssim \left\lVert \sum_{n=1}^N \frac{ e^{2\pi i n^2 x}}{n^{2s}} \right\rVert_4 \lesssim  (\log N)^{3/4}, \qquad s = 1/4,  \\
 \end{equation}
 and 
  \begin{equation}
     \left\lVert \sum_{n=1}^N \frac{ e^{2\pi i n^2 x}}{n^{2s}} \right\rVert_p \simeq \left\{ \begin{array}{ll}
         N^{1 - 2/p - 2s}, & s < 1/2 - 1/p, \\
         1, & s > 1/2-1/p,
     \end{array}
     \right.
     \qquad \text{ for }  p > 4,
 \end{equation}
 
 \begin{equation}
     (\log N )^{1/p} \lesssim \left\lVert \sum_{n=1}^N \frac{ e^{2\pi i n^2 x}}{n^{2s}} \right\rVert_p \lesssim  (\log N)^{1/2}, \qquad s = 1/2 - 1/p, \quad  \text{ for } p > 4,
 \end{equation}

which is equivalent to the statement of the theorem.
\end{proof}

From Theorem~\ref{TheoremNewZalcwasser}, we deduce the following corollary for the integrability of $R_s$.
\begin{cor}
Let $1 \leq p < \infty$ and $R_s$ be defined by \eqref{RiemannGeneralization}. Then,
\begin{itemize}
    \item if $s \leq 1/4$, then $R_s\notin L^p(\mathbb T)$;
    \item if $1/4<s<1/2$, then $R_s\in L^p(\mathbb T)$ if and only if $p<2/(1-2s),$ even if $R_s$ is not absolutely convergent;
    \item  if $s\geq 1/2$, then $R_s\in L^p(\mathbb T)$.
\end{itemize}
\end{cor}

\section{High-pass filters}\label{SECTION_HighPassFilters}

We compute the $L^p$ norms of the high-pass filters of $R_s$, complementing Theorem~\ref{TheoremNewZalcwasser}. Theorem \ref{Theorem_FlatnessHP} follows immediately.

\begin{thm}\label{TheoremHighPassFilters}
Let $0<p<\infty$, $s > 1/4$ if $p \leq 4$, and  $s > 1/2-1/p$ if $p > 4$, and $R_s$ defined by \eqref{RiemannGeneralization}. Then, for large enough $N \in \mathbb N$, we have
\begin{equation}\label{eqthm2}
 \lVert (R_s)_{\geq N} \rVert_{L^p} 
\simeq_{s,p} \begin{cases}
 N^{1/4-s} , & \text{ if } p < 4, \\
N^{1/4-s} \,  (\log N)^{1/4}  , & \text{ if } p = 4, \\
N^{1/2 - 1/p - s}, & \text{ if } p > 4.
\end{cases} 
\end{equation}
The estimate holds for $p=\infty$ interpreting $1/\infty = 0$.
\end{thm}

\begin{proof}
Like in \eqref{LittlewoodPaleyDecompositionInProof}, decompose $R_s$ as
\begin{equation}
    R_s(x) = \sum_{n=1}^\infty \frac{e^{2\pi i n^2 x}}{n^{2s}} = \sum_{k=1}^{\infty} \left( \sum_{n=A^k}^{A^{k+1}-1} \frac{e^{2\pi i n^2 x}}{n^{2s}} \right)  = \sum_{k=1}^{\infty} \Delta_k(x).
\end{equation}
That means that for $N \in \mathbb N$ and for $k(N) \in \mathbb N$ such that $A^{k(N)} \leq N < A^{k(N)+1}$, following the definition of high-pass filters in \eqref{HighPassFilterDefinition} we can write
\begin{equation}\label{LittlewoodPaleyForHighPassFilter}
    (R_s)_{\geq N}(x) = \sum_{n=\sqrt{N}}^\infty \frac{e^{2\pi i n^2 x}}{n^{2s}} = \widetilde{\Delta}_{k(\sqrt{N})}(x) + \sum_{k=k(\sqrt{N})+1}^\infty \Delta_k(x),
\end{equation}
where $\widetilde{\Delta}_{k(\sqrt{N})}(x) = \sum_{n=\sqrt{N}}^{A^{k(\sqrt{N})+1}-1} n^{-2s}e^{2\pi i n^2 x}$. In view of Proposition~\ref{PropDeltak}, given that $A^{k(\sqrt{N})} \simeq \sqrt{N}$ and $k(\sqrt{N}) \simeq \log N$, it will be enough to prove
\begin{equation}\label{HighPassFilterIsLP}
 \lVert (R_s)_{\geq N} \rVert_p  \simeq_p   \lVert \Delta_{k(\sqrt{N})+1} \rVert_p. 
\end{equation}
First, for any $p \geq 1$ the triangle inequality in \eqref{LittlewoodPaleyForHighPassFilter} implies
\begin{equation}\label{SplitInThreeTerms}
 \lVert (R_s)_{\geq N} \rVert_{L^p}  \leq \lVert \widetilde{\Delta}_{k(\sqrt{N})}  \rVert_p + \lVert \Delta_{k(\sqrt{N}) + 1}  \rVert_p + \sum_{k = k(\sqrt{N})+2}^\infty \lVert \Delta_k  \rVert_p.
\end{equation} 
We are going to see that all three terms in \eqref{SplitInThreeTerms} are bounded from above by $\lVert \Delta_{k(\sqrt{N})+1} \rVert_p$.
Since $A^{k(\sqrt{N})+1}/\sqrt{N} \leq A^{k(\sqrt{N})+1}/A^{k(\sqrt{N})} = A$, by Lemma~\ref{lemmaluis}, the definition of $K_N$ in \eqref{ZalcwasserFunction} and Theorem~\ref{TheoremZalcwasser} we can write
\begin{equation}\label{FirstBound}
\begin{split}
    \left\lVert \widetilde{\Delta}_{k(\sqrt{N})}  \right\rVert_p & \simeq_{s,p} \frac{1}{A^{2k(\sqrt{N})s}} \, \left\lVert \sum_{n=\sqrt{N}}^{A^{k(N)+1}-1} e^{2\pi i n^2 x} \right\rVert_p \leq \frac{1}{A^{2k(\sqrt{N})s}} \, \left( \left\lVert K_{A^{k(\sqrt{N})+1}-1} \right\rVert_p + \left\lVert K_{\sqrt{N}} \right\rVert_p \right) \\
    & \lesssim_{s,p} \frac{1}{A^{2k(\sqrt{N})s}} \lVert K_{A^{k(\sqrt{N})}}  \rVert_p  .
\end{split}
\end{equation}
According to Proposition~\ref{PropDeltak}, this is precisely the behavior of $\lVert \Delta_{k(\sqrt{N})+1} \rVert_p$. Regarding the third term in \eqref{SplitInThreeTerms}, we separate cases of $p$:
\begin{itemize}
    \item $1\leq p<4$: by Proposition~\ref{PropDeltak} we know that $\lVert \Delta_k \rVert_p \simeq_p A^{k(1/2-2s)}$, so as long as $s > 1/4$ the series is convergent and we have 
    \begin{equation}
        \sum_{k=k(\sqrt{N}) + 2}^\infty \lVert \Delta_k \rVert_p \simeq_{p}  \sum_{k=k(\sqrt{N}) + 2}^\infty A^{k(1/2-2s)} \simeq_s A^{k(\sqrt{N})(1/2-2s)},
    \end{equation}
    which according to Proposition~\ref{PropDeltak} is like \eqref{FirstBound}.
    
    \item $p>4$: from Proposition~\ref{PropDeltak} we get $\lVert \Delta_k \rVert_p \simeq_p A^{k(1 - 2/p - 2s)}$, so as long as $s > 1/2 - 1/p$ the series is convergent and we have 
     \begin{equation}
        \sum_{k=k(\sqrt{N}) + 2}^\infty \lVert \Delta_k \rVert_p \simeq_{p}  \sum_{k=k(\sqrt{N}) + 2}^\infty A^{k(1 - 2/p - 2s)} \simeq_{p,s} A^{k(\sqrt{N})(1 - 2/p - 2s)},
    \end{equation}
    Again, by Proposition~\ref{PropDeltak} this is the same behavior as in \eqref{FirstBound}.
    
    \item $p=4$: in this case, Proposition~\ref{PropDeltak} gives  $\lVert \Delta_k \rVert_p \simeq_p k^{1/4}\, A^{k(1/2 - 2s)}$, so for $s>1/4$, 
    \begin{equation}
        \sum_{k=k(\sqrt{N}) + 2}^\infty \lVert \Delta_k \rVert_4 \simeq  \sum_{k=k(\sqrt{N}) + 2}^\infty k^{1/4}\, A^{k(1/2 - 2s)} \leq \left( \sum_{k=k(\sqrt{N}) + 2}^\infty k\, r^k \right)^{1/4} \, \left( \sum_{k=k(\sqrt{N}) + 2}^\infty  r^k \right)^{3/4}
    \end{equation}
    where we call $r=A^{1/2 - 2s} < 1$. The second sum is $\simeq r^{k(\sqrt{N})}$. For the first sum, we differentiate the power series and write
    \begin{equation}
    \begin{split}
        \sum_{k=k(\sqrt{N}) + 2}^\infty k\, r^k & \leq \sum_{k=k(\sqrt{N}) + 2}^\infty (k+1)\, r^k = \frac{d}{dr} \sum_{k=k(\sqrt{N}) + 2}^\infty r^k = \frac{d}{dr} \frac{r^{k(\sqrt{N})+2}}{1-r} \\
        & \lesssim_s k(\sqrt{N}) r^{k(\sqrt{N})},
    \end{split}
    \end{equation}
    where the last inequality is possible if $N$ is large enough. Thus, 
    \begin{equation}
        \sum_{k=k(\sqrt{N}) + 2}^\infty \lVert \Delta_k \rVert_4 \lesssim k(\sqrt{N})^{1/4} A^{k(\sqrt{N})(1/2 - 2s)},
    \end{equation}
    which by Theorem~\ref{TheoremZalcwasser} is the same as \eqref{FirstBound}.
\end{itemize}

Regarding the lower bound, if $p>1$ the Littlewood-Paley theorem directly gives 
\begin{equation}\label{LowerBoundWithLP}
    \lVert (R_s)_{\geq N} \rVert_{L^p} \simeq_p \left\lVert \left( \left| \widetilde{\Delta}_{k(\sqrt{N})}  \right|^2 +  \sum_{k = k(\sqrt{N}) + 1}^\infty{\left| \Delta_k \right|^2} \right)^{1/2}  \right\rVert_{L^p} \geq \left\lVert  \Delta_{k(\sqrt{N})+1}   \right\rVert_{L^p},
\end{equation}
thus proving \eqref{HighPassFilterIsLP} and the theorem for every $p>1$. To get the result for $p=1$, and also for $0<p<1$, one can use H\"older's inequality and proceed like in \eqref{ExtrapolationByHolder}, \eqref{Interpolation1} and \eqref{Interpolation2}. \end{proof}

\section{Structure Functions and \texorpdfstring{$R_s$}{Rs}}\label{SECTION_StructureFunctions}

To prove Theorem \ref{Theorem_FlatnessSF} it suffices to compute the asymptotic behavior of the structure functions of $R_s$. For simplicity, let us denote $S_{R_s,p}(\ell)$ simply by $S_{s,p}(\ell)$. To compare with the results for high-pass filters in Theorem~\ref{TheoremHighPassFilters}, we write the theorem for $S_{s,p}^{1/p}$ which scales as $\lVert (R_s)_{\geq N} \rVert_p$.

\begin{thm}\label{TheoremStructureFunctions}
Let $0 < p < \infty$. Let $s > 1/4$ if $p \leq 4$ and $s> 1/2 - 1/p$ if $p>4$. For $0 < \ell \ll 1 $ small enough, 
\begin{itemize}
    \item if $s < 5/4$, 
        \begin{equation}
            S_{s,p}(\ell)^{1/p} \simeq \left\{ \begin{array}{ll}
                \ell^{s - 1/4 }, & p < 4,  \\
                \ell^{s - 1/4} \, \log(\ell^{-1})^{1/4}, & p = 4, \\
                \ell^{s + 1/p - 1/2}, & p > 4.
            \end{array}
            \right.
        \end{equation}
        
    \item if $s=5/4$, 
        \begin{equation}
            S_{5/4,p}(\ell)^{1/p} \simeq \left\{ \begin{array}{ll}
                \ell \, \log(\ell^{-1})^{1/2}, & p < 4,  \\
                \ell^{3/4 + 1/p}, & p > 4,
            \end{array}
            \right.
        \end{equation}
        and when $p=4$,
        \begin{equation}
            \ell\, \log(\ell^{-1})^{1/2} \lesssim S_{5/4,4}(\ell)^{1/4} \lesssim \ell\,\log(\ell^{-1})^{3/4}.
        \end{equation}
        
    \item if $5/4 < s < 3/2$:
        \begin{equation}
            S_{s,p}(\ell)^{1/p} \simeq \left\{ \begin{array}{ll}
                \ell , & p < 2/(3-2s),  \\
                \ell^{s + 1/p - 1/2}, & p > 2/(3-2s),
            \end{array}
            \right.
        \end{equation}
        and when $p=2/(3-2s)$,
        \begin{equation}
            \ell \lesssim S_{s,p}(\ell)^{1/p} \lesssim \ell\, \log(\ell^{-1})^{1/2}, \qquad p = \frac{2}{3-2s}.
        \end{equation}
        Moreover, if $p=2/(3-2s)$ is an even integer, then we have 
        \begin{equation}
            \ell\log(\ell)^{1/p} \lesssim S_{s,p}(\ell)^{1/p} \lesssim \ell\, \log(\ell^{-1})^{1/2}.
        \end{equation}
    
    \item if $s \geq 3/2$:
        \begin{equation}
            S_{s,p}(\ell)^{1/p} \simeq \ell, \qquad \forall p > 0.
        \end{equation}

\end{itemize}
\end{thm}

In the rest of the section we prove Theorem~\ref{TheoremStructureFunctions}. First, using the Plancherel theorem, we compute the asymptotic behavior for even $p$. Then, using H\"older's inequality, we solve the case $s\geq3/2$ for all $p$.
For $s<3/2$, we first compute the upper bounds by using Lemma \ref{Application2} and Theorem \ref{TheoremNewZalcwasser}. Finally, using those upper bounds and the cases when $p \in 2\mathbb N$, we obtain the lower bounds interpolating with H\"older's inequality.

Recalling the definition of structure functions in \eqref{StructureFunction}, we first write
\begin{equation}
    S_{s,p}(\ell) ^{1/p}
    = \left( \int_{\mathbb T} \left| R_s(x+\ell/2) - R_s(x-\ell/2) \right|^p\, dx \right)^{1/p}
    = \lVert I_s(\cdot, \ell) \rVert_p,
\end{equation}
where
\begin{equation}\label{DefI_s}
    I_s(x,\ell) =  R_s(x+\ell/2) - R_s(x-\ell/2) = 2i\, \sum_{n=1}^\infty \frac{\sin (\pi n^2 \ell)}{n^{2s}}\, e^{2\pi i n^2 x}.
\end{equation}
From now on, we denote \eqref{DefI_s} just by $I_s(\ell)$, and we drop the contribution of the factor $2i$. The main idea to obtain upper bounds for $S_{s,p}$ is to split $I_s$ in high and low frequencies, that is, for some $M > 1$, 
\begin{equation}\label{UpperAndLowerFilters}
    I_s(\ell)  = \sum_{n=1}^{M^{1/2}} \frac{\sin (\pi n^2 \ell)}{n^{2s}}\, e^{2\pi i n^2 x} + \sum_{n>M^{1/2}} \frac{\sin (\pi n^2 \ell)}{n^{2s}}\, e^{2\pi i n^2 x}  =  \left( I_s\right)_{\leq M}(\ell)   + \left(  I_s \right)_{>M} (\ell),
\end{equation}
for which we used the notation of filters in \eqref{HighPassFilterDefinition}. The triangle inequality immediately gives the upper bound
\begin{equation}\label{UpperBoundWithPlancherel}
    S_{s,p}(\ell)^{1/p} = \lVert I_s(\ell) \rVert_p \leq  \lVert \left( I_s\right)_{\leq M}(\ell)  \rVert_p +  \lVert \left( I_{s}\right)_{> M}(\ell)  \rVert_p, \qquad \forall M > 1, \quad \forall p \geq 1.
\end{equation}
Moreover, by \eqref{DefI_s} and Theorem~\ref{TheoremHighPassFilters} we can bound the high-pass filter by
\begin{equation}\label{BoundingHighPassFilter}
    \lVert \left( I_s\right)_{> M}(\ell)  \rVert_p 
     \leq \lVert \left( R_s(\cdot + \ell/2)  \right)_{> M} \rVert_p  + \lVert \left(  R_s(\cdot - \ell/2)  \right)_{> M} \rVert_p  = 2\, \lVert \left( R_s  \right)_{> M} \rVert_p, \quad  \forall p \geq 1. 
\end{equation}
Lower bounds, on the contrary, are more delicate. By choosing $M = 1/(p\ell)$ and using the Plancherel theorem, we will show
\begin{equation}
    S_{s,p}(\ell)^{1/p} \geq \lVert \left( I_s\right)_{\leq (p\ell)^{-1}}(\ell)  \rVert_p, \qquad \forall p \in 2\mathbb N,
\end{equation}
which will allow us to compute $S_{s,p}(\ell)$ for $p \in 2\mathbb N$. The lower bounds for the rest of $p$ will follow by interpolation from $p \in 2\mathbb N$ and the upper bounds from \eqref{UpperBoundWithPlancherel}.

\subsection{Upper and lower bounds for \texorpdfstring{$p \in 2 \mathbb N$}{p} and every \texorpdfstring{$s$}{s}}

Let $p = 2k$ with $k \in \mathbb N$. Split $I_{s}(\ell)$ like in \eqref{UpperAndLowerFilters} with $M = (p\ell)^{-1}$. We are going to show that  
\begin{equation}\label{LowerBoundWithPlancherel}
    \lVert \left( I_s\right)_{\leq (p\ell)^{-1}}(\ell)  \rVert_p \leq S_{s,p}(\ell)^{1/p}  \leq  \lVert \left( I_s\right)_{\leq (p\ell)^{-1}}(\ell)  \rVert_p +  \lVert \left( I_{s}\right)_{> (p\ell)^{-1}}(\ell)  \rVert_p.
\end{equation}
Like in \eqref{UpperBoundWithPlancherel}, the upper bound comes from the triangle inequality. We prove the lower bound and compute $\lVert \left( I_s\right)_{\leq (p\ell)^{-1}}(\ell)  \rVert_p$ using the Plancherel theorem. Call
\begin{equation}
    a_n = \frac{\sin (\pi n \ell)}{n^{s}}, \qquad I_s(\ell) = \sum_{n=1}^\infty a_{n^2}\, e_{n^2} 
\end{equation}
so that 
\begin{equation}\label{CompleteSeries}
    \left\lVert I_s(\ell) \right\rVert_p^p = \left\lVert \left( \sum_{n=1}^\infty a_{n^2}\, e_{n^2} \right)^k  \right\rVert_2^2 = \sum_{n=1}^\infty |b_n|^2, \qquad \text{ where } \quad
    b_n = \sum_{n_1^2+\ldots + n_k^2 = n} a_{n_1^2} \ldots a_{n_k^2}.
\end{equation}
In the same way, let 
\begin{equation}\label{FilterSeries}
    \left\lVert (I_s)_{\leq (p\ell)^{-1}}(\ell) \right\rVert_p^p 
    = \sum_{n=1}^{k/(p\ell)} |\tilde b_n|^2, \qquad \text{ where } \quad 
    \tilde b_n = \sum_{\substack{n_1^2 + \ldots + n_k^2 = n \\ n_i  \leq (p\ell)^{-1/2}, \forall i}} a_{n_1^2} \ldots a_{n_k^2}.
\end{equation}
Observe that
\begin{equation}\label{SineIsBounded}
    \frac12 \leq \frac{\sin(\pi n^2 \ell)}{ \pi n^2 \ell} \leq 1, \qquad \text{ whenever } \quad  n^2 \leq \frac{k}{p\ell} = \frac{1}{2\ell},
\end{equation}
which implies that $a_{n^2}>0$
for every $n^2 \leq k/(p\ell)$. Thus, 
\begin{equation}
    \tilde b_n = \sum_{\substack{n_1^2 + \ldots + n_k^2 = n \\ n_i \leq (p\ell)^{-1/2}, \forall i}} a_{n_1^2} \ldots a_{n_k^2} \leq \sum_{n_1^2 + \ldots + n_k^2 = n } a_{n_1^2} \ldots a_{n_k^2} = b_n, \quad \forall n \leq k/(p\ell),
\end{equation}
and hence
\begin{equation}
    \left\lVert (I_s)_{\leq (p\ell)^{-1}}(\ell) \right\rVert_p^p = \sum_{n=1}^{k/(p\ell)} \tilde b_n^2 \leq \sum_{n=1}^{k/(p\ell)} b_n^2 \leq \left\lVert I_s(\ell) \right\rVert_p^p,
\end{equation}
which proves \eqref{LowerBoundWithPlancherel}.
What is more, from \eqref{SineIsBounded} we have $a_{n^2} \simeq \ell /n^{2(s-1)}$ for every $n^2 \leq k/(p\ell)$, so in view of \eqref{CompleteSeries} and  \eqref{FilterSeries}, we have
\begin{equation}\label{LowPassFilterToNewZalcwasser}
    \lVert \left( I_{s}\right)_{\leq (p\ell)^{-1}}(\ell)  \rVert_p^p =  \left\lVert \sum_{n=1}^{(p\ell)^{-1/2}} a_{n^2}\,e_{n^2}  \right\rVert_p^p \simeq  \ell^p \, \left\lVert  \sum_{n=1}^{(p\ell)^{-1/2}}  \frac{e_{n^2}}{n^{2(s-1)}}  \right\rVert_p^p,
\end{equation}
which we know from Theorem~\ref{TheoremNewZalcwasser}:

\begin{itemize}
    \item For $p=2$, Theorem~\ref{TheoremNewZalcwasser} gives
    \begin{equation}
        \lVert \left( I_{s}\right)_{\leq (2\ell)^{-1}}(\ell)  \rVert_2 \simeq \left\{ \begin{array}{ll}
            \ell^{s-1/4}, & \text{ if } s < 5/4, \\
            \ell \, \log(\ell^{-1})^{1/2}, & \text{ if } s=5/4, \\
            \ell & \text{ if } s > 5/4.
        \end{array}
        \right.
    \end{equation}
    Since by Theorem~\ref{TheoremHighPassFilters} we have $\lVert \left( R_s  \right)_{> (2\ell)^{-1}} \rVert_2 \simeq \ell^{s-1/4}$ for every $s > 1/4$, in view of \eqref{LowerBoundWithPlancherel} and  \eqref{BoundingHighPassFilter} we get
    \begin{equation}\label{SFwhenPis2}
        S_{s,2}(\ell)^{1/2} \simeq \lVert \left( I_{s}\right)_{\leq (2\ell)^{-1}}(\ell)  \rVert_2 \simeq \left\{ \begin{array}{ll}
            \ell^{s-1/4}, & \text{ if } s < 5/4, \\
            \ell \, \log(\ell^{-1})^{1/2}, & \text{ if } s=5/4, \\
            \ell & \text{ if } s > 5/4.
        \end{array}
        \right.
    \end{equation}

    \item For $p=4$, Theorem~\ref{TheoremNewZalcwasser} gives
    \begin{equation}
        \lVert \left( I_{s}\right)_{\leq (4\ell)^{-1}}(\ell)  \rVert_4  \simeq_s 
        \left\{ \begin{array}{ll}
            \ell^{s - 1/4}\, \log(\ell^{-1})^{1/4} & \text{ if } s < 5/4, \\
            \ell & \text{ if } s > 5/4,
        \end{array}
        \right.
    \end{equation}
    while Theorem~\ref{TheoremHighPassFilters} and  \eqref{BoundingHighPassFilter} imply that
    \begin{equation}\label{SFHighPassFilterWhenP4}
        \lVert \left( I_s\right)_{> (4\ell)^{-1}}(\ell)  \rVert_4 
        \lesssim_s \ell^{s - 1/4}\, \log(\ell^{-1})^{1/4}, \qquad \forall s > 1/4.
    \end{equation}
    Since $s > 5/4$ implies $\ell^{s-1/4} \ll \ell$ when $\ell \ll 1$, from \eqref{LowerBoundWithPlancherel} we get
    \begin{equation}
    S_{s,4}(\ell)^{1/4} \simeq_s \lVert \left( I_{s}\right)_{\leq (4\ell)^{-1}}(\ell)  \rVert_4  \simeq_s
       \left\{ \begin{array}{ll}
            \ell^{s-1/4}\, \log(\ell^{-1})^{1/4} & \text{ if } s < 5/4, \\
            \ell & \text{ if } s > 5/4. 
        \end{array}
        \right.
    \end{equation}
    Regarding $s=5/4$, \eqref{LowPassFilterToNewZalcwasser} and Theorem~\ref{TheoremNewZalcwasser} give
    \begin{equation}
        \ell\, \log(\ell^{-1})^{1/2} \lesssim \left\lVert ( I_{5/4} )_{\leq (4\ell)^{-1}}(\ell)  \right\rVert_4 \lesssim \ell\,  \log(\ell^{-1})^{3/4}, 
    \end{equation}
    and since from \eqref{SFHighPassFilterWhenP4} we have $ \lVert \left( I_s\right)_{> (4\ell)^{-1}}(\ell)  \rVert_4 \lesssim \ell\, \log(\ell^{-1})^{1/4} $, we get
    \begin{equation}
        \ell\, \log(\ell^{-1})^{1/2} \lesssim S_{5/4, 4}(\ell)^{1/4} \lesssim \ell\, \log(\ell^{-1})^{3/4}. 
    \end{equation}

    \item For $p =6, 8, 10, \ldots$, \eqref{LowPassFilterToNewZalcwasser} and Theorem~\ref{TheoremNewZalcwasser} give
    \begin{equation}
        \lVert \left( I_{s}\right)_{\leq (p\ell)^{-1}}(\ell)  \rVert_p 
        \simeq_{s,p} 
        \left\{ \begin{array}{ll}
            \ell^{s  + 1/p - 1/2 }, & \text{ if } s < 3/2 - 1/p, \\
            \ell, & \text{ if } s > 3/2 - 1/p,
        \end{array}
        \right.
    \end{equation}
    and from \eqref{BoundingHighPassFilter} and  Theorem~\ref{TheoremHighPassFilters} we get
    \begin{equation}\label{SFHighPassFilterWhenPGreaterThan4}
        \lVert \left( I_s\right)_{> (p\ell)^{-1}}(\ell)  \rVert_p \lesssim_{s,p}
        \ell^{s + 1/p  - 1/2}, \qquad \forall s > 1/2 - 1/p.
    \end{equation}
    Thus, since $s > 3/2 - 1/p$ implies $\ell^{s + 1/p - 1/2} \ll \ell$ when $\ell \ll 1$, from \eqref{LowerBoundWithPlancherel} we get
    \begin{equation}
    S_{s,p}(\ell)^{1/p} \simeq_{s,p} \lVert \left( I_{s}\right)_{\leq (p\ell)^{-1}}(\ell)  \rVert_p  \simeq_{s,p}
       \left\{ \begin{array}{ll}
            \ell^{s + 1/p - 1/2}, & \text{ if } s < 3/2 - 1/p, \\
            \ell, & \text{ if } s > 3/2 - 1/p. 
        \end{array}
        \right.
    \end{equation}
    When $s = 3/2 - 1/p$, \eqref{LowPassFilterToNewZalcwasser} and Theorem~\ref{TheoremNewZalcwasser} gives
    \begin{equation}
        \ell \, \log(\ell^{-1})^{1/p} \lesssim_p \left\lVert  ( I_{3/2-1/p} )_{\leq (p\ell)^{-1}}(\ell)  \right\rVert_p \lesssim_p \ell\,  \log(\ell^{-1})^{1/2}, 
    \end{equation}
    and from \eqref{SFHighPassFilterWhenPGreaterThan4} we have $ \lVert ( I_{3/2-1/p} )_{> (p\ell)^{-1}}(\ell)  \rVert_p \lesssim \ell $, so 
    \begin{equation}\label{even p and critical s}
        \ell \, \log(\ell^{-1})^{1/p} \lesssim_p S_{3/2-1/p,p}(\ell)^{1/p} \lesssim_p \ell\, \log(\ell^{-1})^{1/2}. 
    \end{equation}
\end{itemize}

Thus the theorem is established for $p \in 2\mathbb N$ and for all $s$.

\subsection{Upper and lower bounds for \texorpdfstring{$s \geq 3/2$}{s32}}

When $s \geq 3/2$, we just saw that $S_{s,p}(\ell)^{1/p} \simeq \ell$ for every $p \in 2\mathbb N$. 
Thus, the result for general $p$ follows by interpolation in a very similar way as we did in \eqref{ExtrapolationByHolder}, \eqref{Interpolation1} and \eqref{Interpolation2}. 

Let $k \in \mathbb N$ such that  $2k < p < 2k+2$. For the upper bound, we interpolate for $p$ between $2k$ and $2k+2$ with H\"older's inequality by writing
\begin{equation}
    S_{s,p}(\ell) = \int |I_s(\ell)|^p  \leq \left( \int|I_s(\ell)|^{a\theta} \right)^{1/\theta} \, \left( \int |I_s(\ell)|^{b\theta'} \right)^{1/\theta'},
\end{equation}
where $a+b = p$, $a\theta = 2k$, $b\theta' = 2k+2$, $1/\theta + 1/\theta' = 1$. This implies $\theta = 2/(2k+2-p)>1$. Thus, since we know $S_{s,2k}$ and $S_{s,2k+2}$ when $k\geq 1$, we get
\begin{equation}
    S_{s,p}(\ell) \leq S_{s,2k}(\ell)^{1/\theta}\, S_{s,2k+2}(\ell)^{1/\theta'} \simeq_{s,p} \ell^{2k(2k+2-p)/2}\, \ell^{(2k+2)(p-2k)/2} = \ell^p
\end{equation}
for all $p > 2$. Like in \eqref{ExtrapolationByHolder}, the case $p<2$ easily follows from H\"older's inequality because
\begin{equation}\label{RepetitionOfExtrapolation}
    S_{s,p}(\ell) = \int |I_s(\ell)|^p \leq \left( \int |I_s(\ell)|^{p \, \frac{2}{p}}  \right)^{p/2} = S_{s,2}(\ell)^{p/2} \simeq_s  \ell^p. 
\end{equation}
For the lower bound, similar to what we did in \eqref{Interpolation1} and \eqref{Interpolation2}, we interpolate for $2k+2$ between $p$ and $2k+4$, that is, 
\begin{equation}\label{InterpolationGeneralWithHolder}
    S_{s,2k+2}(\ell) = \int |I_s(\ell)|^{2k+2} \leq \left( \int |I_s(\ell)|^{a\theta}  \right)^{1/\theta} \, \left( \int |I_s(\ell)|^{b\theta'}  \right)^{1/\theta'},
\end{equation}
where now $2k+2 = a+b$, $a\theta = p$, $b\theta' = 2k+4$, and also $1/\theta + 1/\theta'=1$. This implies $\theta = (2k+4-p)/2 > 1$. Then, 
\begin{equation}\label{InterpolationGeneral}
    S_{s,2k+2}(\ell) \leq S_{s,p}(\ell)^{1/\theta}\, S_{s,2k+4}^{1/\theta'} \quad \Longrightarrow \quad S_{s,p}(\ell) \geq S_{s,2k+2}^{\theta}\, S_{s,2k+4}^{-\theta / \theta'},
\end{equation}
and since we know the behavior of $S_{s,2k+2}$ and $S_{s,2k+4}$ as long as $k \geq 0$, we get 
\begin{equation}
    S_{s,p}(\ell) \geq \ell^{(2k+2)\theta} \, \ell^{(2k+4)(1-\theta)} = \ell^{2k+4 - 2\theta} = \ell^p, \qquad \forall p > 0.
\end{equation}
In short, we have proved that $S_{s,p}(\ell)^{1/p} \simeq_{s,p} \ell$ for every $p > 0$, concluding the case $s\geq 3/2$.

\begin{rmk}
Interpolation works for $s \geq 3/2$ because we have the same expression for all $p$. That does not happen when $s < 3/2$, where different ranges of $p$ give different expressions for $S_{s,p}$. Moreover, either in $p=4$ or in $p = 2/(3-2s)$ extra logarithmic terms appear. To solve this problem, we first compute the upper bounds using \eqref{UpperBoundWithPlancherel}.
\end{rmk}

\subsection{Upper bounds for \texorpdfstring{$s<3/2$}{s32} and all \texorpdfstring{$p$}{p}} \label{Subsection_SFUpperBoundsByDirectComputation}
We work with \eqref{UpperBoundWithPlancherel} with $M = \ell^{-1}$. Assume first that $p \geq 1$. From \eqref{BoundingHighPassFilter} and Theorem~\ref{TheoremHighPassFilters} we directly get
\begin{equation}\label{SFHighPassFilterUpperBound}
     \lVert \left(I_s\right)_{> \ell^{-1}}(\ell) \rVert_p \lesssim_{s,p} \left\{
    \begin{array}{ll}
         \ell^{s - 1/4}, & \text{ if } p < 4,   \\
        \ell^{s - 1/4}\, \log(\ell^{-1})^{1/4},  & \text{ if } p=4, \\
        \ell^{s + 1/p - 1/2}, & \text{ if } p >4.
    \end{array}
    \right.
\end{equation}
To deal with $\lVert \left(I_s\right)_{\leq \ell^{-1}} (\ell) \rVert_p$, we use Lemma~\ref{Application2} with $B=\ell^{-1}$. 
We have two cases:
\begin{itemize}
    \item For $s \leq 1$, it gives
    \begin{equation}
    \lVert \left(I_s\right)_{\leq \ell^{-1}} (\ell) \rVert_p  
    = 
    \left\lVert \sum_{n=1}^{\ell^{-1/2}} \frac{\sin(\pi n^2 \ell)}{n^{2s}} \, e_{n^2} \right\rVert_p 
    \lesssim_s  
    \ell \, \Bigg\lVert \sum_{n=1}^{\ell^{-1/2}} n^{2(1-s)} \, e_{n^2} \Bigg\rVert_p
    +
    \ell^{s} \,  \Bigg\| \sum_{n=1}^{\ell^{-1/2}} e_{n^2} \Bigg\|_{L^p}.
    \end{equation}
    Using Theorem~\ref{TheoremNewZalcwasser} to bound the $L^p$ norms, 
    we get the same bound as for $\left(I_s\right)_{> \ell^{-1}}$ in \eqref{SFHighPassFilterUpperBound}. Thus, 
    \begin{equation}\label{SFLowPassFilterUpperBound}
     S_{s,p}(\ell)^{1/p}  \lesssim_{s,p} \left\{
    \begin{array}{ll}
         \ell^{ s - 1/4}, & \text{ if } p < 4,   \\
        \ell^{s - 1/4}\, \log(\ell^{-1})^{1/4},  & \text{ if } p=4, \\
        \ell^{s + 1/p - 1/2}, & \text{ if } p >4,
    \end{array}
    \right.
    \qquad \text{ for } s \leq 1.
    \end{equation}
    
    \item For $1<s<3/2$, Lemma~\ref{Application2} gives
        \begin{equation}
            \lVert \left(I_s\right)_{\leq \ell^{-1}} (\ell) \rVert_p  
            \lesssim_s 
            \ell \, \Bigg\| \sum_{n=1}^{\ell^{-1/2}} \,  \frac{e_{n^2}}{n^{2(s-1)}} \Bigg\|_{L^p} 
            +
            \ell^3 \, \Bigg\| \sum_{n=1}^{\ell^{-1/2}} \,  n^{2(3-s)}\, e_{n^2} \Bigg\|_{L^p} 
            +
            \ell^{s} \, \Bigg\| \sum_{n=1}^{\ell^{-1/2}}  e_{n^2} \Bigg\|_{L^p}.
        \end{equation}
        We further split cases:
        \begin{itemize}
            \item $1 < s < 5/4$: When $p < 4$, from Theorem~\ref{TheoremNewZalcwasser} we get
            \begin{equation}
                \lVert \left(I_s\right)_{\leq \ell^{-1}} (\ell) \rVert_p 
                \lesssim_{p,s} 
                \ell \, \ell^{s-5/4} +  \ell^3 \, \ell^{-1/4 + (s-3)}  + \ell^{s} \, \ell^{-1/4} 
                \simeq
                \ell^{ s - 1/4}.
            \end{equation}
            When $p = 4$, we get 
            \begin{equation}
            \begin{split}
                \lVert \left(I_s\right)_{\leq \ell^{-1}} (\ell) \rVert_4 
                & \lesssim_{s} 
                \left(\ell \, \ell^{s-5/4} +  \ell^3 \, \ell^{-1/4 + (s-3)}  + \ell^{s} \, \ell^{-1/4} \right) \, \log(\ell^{-1})^{1/4}
                 \\
                &
                \simeq 
                \ell^{s - 1/4}\, \log(\ell^{-1})^{1/4}.
            \end{split}
            \end{equation}
            When $p > 4$, we get 
            \begin{equation}
            \begin{split}
                \lVert \left(I_s\right)_{\leq \ell^{-1}} (\ell) \rVert_p 
                & \lesssim_{p,s} 
                \ell \, \ell^{s - 3/2 + 1/p} +   \ell^3 \, \ell^{-1/2 + 1/p + (s-3)}   + \ell^{s} \, \ell^{-1/2+1/p} \\
                & \simeq 
                \ell^{s + 1/p - 1/2}.
            \end{split}
            \end{equation}
            The bounds are the same as in \eqref{SFHighPassFilterUpperBound}, so the behavior matches the case $s \leq 1$:
            \begin{equation}\label{SFUpperBoundSSmaller14}
                S_{s,p}(\ell)^{1/p}  \lesssim_{s,p} \left\{
                \begin{array}{ll}
                    \ell^{s - 1/4}, & \text{ if } p < 4,   \\
                    \ell^{s - 1/4}\, \log(\ell^{-1})^{1/4},  & \text{ if } p = 4, \\
                    \ell^{s + 1/p - 1/2}, & \text{ if } p > 4,
                \end{array}
                \right.
                \qquad \text{ for } s < 5/4.
            \end{equation}
            
            \item $s=5/4$:  When $p < 4$, from Theorem~\ref{TheoremNewZalcwasser} we get
            \begin{equation}
                \lVert \left(I_s\right)_{\leq \ell^{-1}} (\ell) \rVert_p 
                \lesssim_{p}
                \ell \, \log(\ell^{-1})^{1/2} +  \ell^3 \, \ell^{ -1/4 + (5/4-3) } + \ell^{5/4} \, \ell^{-1/4}
                \simeq 
                \ell \, \log(\ell^{-1})^{1/2}.
            \end{equation}
            When $p = 4$, we get 
            \begin{equation}
            \begin{split}
                \lVert \left(I_s\right)_{\leq \ell^{-1}} (\ell) \rVert_4 
                & \lesssim
                \ell\,   \log(\ell^{-1})^{3/4} +  \ell^3 \, \ell^{ -1/4 + (5/4-3)} \, \log(\ell^{-1})^{1/4}  +  \ell^{5/4}\, \ell^{-1/4}\, \log(\ell^{-1})^{1/4} \\
                & \simeq
                \ell \, \log(\ell^{-1})^{3/4}.
            \end{split}
            \end{equation}
            When $p > 4$, we get 
            \begin{equation}
                \lVert \left(I_s\right)_{\leq \ell^{-1}} (\ell) \rVert_p 
                \lesssim_{p}
                \ell \, \ell^{-1/4 + 1/p} +  \ell^3 \, \ell^{-1/2 + 1/p + (5/4-3)}  +  \ell^{5/4} \, \ell^{-1/2+1/p} 
                \simeq
                \ell^{3/4 + 1/p }.
            \end{equation}
            The bounds for the high-pass filter in \eqref{SFHighPassFilterUpperBound} are smaller in all cases, so we get 
            \begin{equation}\label{SFUpperBoundS14}
                S_{s,p}(\ell)^{1/p} \lesssim_p \left\{ 
                \begin{array}{ll}
                    \ell \, \log(\ell^{-1})^{1/2} & \text{ if } p < 4,  \\
                    \ell \, \log(\ell^{-1})^{3/4} & \text{ if } p = 4, \\
                    \ell^{3/4 + 1/p }, & \text{ if } p > 4.
                \end{array}
                \right.
                \qquad \text{ when } s = 5/4. 
            \end{equation}
            
            \item $5/4 < s < 3/2$: In this case,  for $p < 4$ Theorem~\ref{TheoremNewZalcwasser} gives 
            \begin{equation}
                \lVert \left(I_s\right)_{\leq \ell^{-1}} (\ell) \rVert_p 
                \lesssim_{s,p}
                \ell +  \ell^3 \, \ell^{-1/4 + (s-3)}  +  \ell^{s} \, \ell^{-1/4} 
                \simeq 
                \ell + \ell^{ s - 1/4} \simeq \ell,
            \end{equation}
            because $ s - 1/4 > 1$. 
            When $p = 4$, we get 
            \begin{equation}
                \lVert \left(I_s\right)_{\leq \ell^{-1}} (\ell) \rVert_4 
                \lesssim_s
                \ell +  \ell^3 \, \ell^{-1/4 + (s-3)} \, \log{(\ell^{-1})}^{1/4}  + \ell^{s}\, \ell^{-1/4}\, \log(\ell^{-1})^{1/4} 
                \simeq
                \ell.
            \end{equation}
            When $4 < p < 2/(3-2s)$, which corresponds to $s > 3/2 - 1/p$, we get 
            \begin{equation}
                \lVert \left(I_s\right)_{\leq \ell^{-1}} (\ell) \rVert_p 
                \lesssim_{s,p}
                \ell +  \ell^3 \, \ell^{-1/2 + 1/p + (s-3)}  + \ell^{s} \, \ell^{-1/2+1/p} 
                \simeq
                \ell + \ell^{s + 1/p - 1/2}
                \simeq
                \ell
            \end{equation}
            because $s + 1/p - 1/2 > 1$. When $p = 2/(3-2s)$, 
            \begin{equation}
                \lVert \left(I_s\right)_{\leq \ell^{-1}} (\ell) \rVert_p 
                \lesssim_{s,p} 
                \ell \, \log(\ell^{-1})^{1/2} +  \ell^3 \, \ell^{-1/2 + 1/p + (s-3)}   +  \ell^{s} \, \ell^{-1/2 + 1/p}
                \simeq 
                \ell\, \log(\ell^{-1})^{1/2}. 
            \end{equation}
            Finally, when $p > 2/(3-2s)$, 
            \begin{equation}
                \lVert \left(I_s\right)_{\leq \ell^{-1}} (\ell) \rVert_p 
                \lesssim_{s,p}
                \ell \, \ell^{s - 3/2 + 1/p} +  \ell^3 \, \ell^{-1/2 + 1/p + (s-3)}  +  \ell^{s} \, \ell^{-1/2 + 1/p} 
                \simeq 
                \ell^{s + 1/p - 1/2} . 
            \end{equation}
            The bounds for the high-pass filter in \eqref{SFHighPassFilterUpperBound} are smaller in all cases, so 
            \begin{equation}\label{SFUpperBoundsSBetween14And12}
                S_{s,p}(\ell)^{1/p} \lesssim_{s,p} \left\{ 
                \begin{array}{ll}
                    \ell & \text{ if } p < 2/(1-2s),  \\
                    \ell \, \log(\ell^{-1})^{1/2} & \text{ if } p = 2/(3-2s), \\
                    \ell^{s + 1/p - 1/2 }, & \text{ if } p > 2/(3-2s).
                \end{array}
                \right.
                \qquad \text{ when } 5/4 < s < 3/2. 
            \end{equation}
        \end{itemize}
\end{itemize}
Upper bounds are thus established for all $s< 3/2$ and $p \geq 1$. For $0<p<1$, like in \eqref{RepetitionOfExtrapolation} H\"older's inequality implies $S_{s,p}(\ell)^{1/p} \leq S_{s,2}(\ell)^{1/2}$, which gives the desired bound for all $s < 3/2$.

\subsection{Lower bounds for all \texorpdfstring{$s < 3/2$}{s32} and all \texorpdfstring{$p$}{p}}

Let $p>0$ not even, so there exists $k \in \mathbb N \cup\{0\}$ such that $2k < p < 2k+2$. We separate cases in $s$:  

\begin{itemize}
    \item $s < 5/4$: Like in \eqref{InterpolationGeneralWithHolder} and \eqref{InterpolationGeneral}, since we know $S_{s,2k+2}$ and $S_{s,2k+4}$, we interpolate for $2k+2$ between $p$ and $2k+4$. However, to avoid the logarithm when $2k+2=4$ in \eqref{SFLowPassFilterUpperBound} and \eqref{SFUpperBoundSSmaller14}, we work with $2k+2 \geq 6$, so $p > 2k \geq 4$. In this case, with $\theta = (2k+4-p)/2$, 
we get
    \begin{equation}\label{InterpolationBetweenEvenCases}
        \begin{split}
        S_{s,p}(\ell) & \geq S_{s,2k+2}^{\theta}\, S_{s,2k+4}^{-\theta/\theta'} \simeq \ell^{(1-(2k+2)/2 + (2k+2)s)\theta}\, \ell^{(1-(2k+4)/2 + (2k+4)s)\,(1-\theta)} \\
        & 
        = \ell^{ps + 1 - p/2}, \qquad \qquad \qquad  \text{ for } p > 4,
        \end{split}
    \end{equation}
    which matches the upper bound in \eqref{SFUpperBoundSSmaller14}. For $2 < p < 4$,  H\"older's inequality gives
    \begin{equation}\label{TrivialHolderForLowerBound}
        S_{s,2}(\ell) = \int |I_s(\ell)|^2 \leq \left( \int |I_s(\ell)|^{2 \, \frac{p}{2}} \right)^{2/p} 
        = S_{s,p}(\ell)^{2/p},
    \end{equation}
    and hence
    \begin{equation}
        S_{s,p}(\ell)^{1/p} \geq S_{s,2}(\ell)^{1/2} \simeq_s \ell^{s-1/4}, \qquad 2 < p < 4.
    \end{equation}
    This matches the upper bound in \eqref{SFUpperBoundSSmaller14}. 
    For $p<2$, we interpolate for 2 between $p$ and 3 like in \eqref{Interpolation1} and \eqref{Interpolation2}. Indeed, with $\theta=3-p$, using \eqref{SFwhenPis2} and the upper bound for $S_{s,3}(\ell)$ in \eqref{SFUpperBoundSSmaller14}, we get
    \begin{equation}\label{InterpolatingBetween2And3_Version2}
        S_{s,p}(\ell) \geq \frac{S_{s,2}(\ell)^\theta}{S_{s,3}(\ell)^{\theta/\theta'}} \gtrsim \frac{\ell^{(2s-1/2)\theta}}{\ell^{(3s-3/4)(\theta - 1)}} = \ell^{ps - p/4}, \qquad 0 < p < 2.
    \end{equation}
    This coincides with the upper bound in \eqref{SFUpperBoundSSmaller14}, so we conclude the case $s<5/4$.

    \item $s=5/4$: 
    Like in \eqref{InterpolationBetweenEvenCases}, for $p>4$ we get
    \begin{equation}
        \begin{split}
        S_{5/4,p}(\ell) & \geq S_{5/4,2k+2}^\theta \, S_{5/4,2k+4}(\ell)^{-\theta/\theta'} \simeq_p \ell^{(1 + 3(k+1)/2)\theta}\, \ell^{(1 + 3(k+2)/2)(1-\theta)} \\
        & = \ell^{ -3\theta/2 + 1 + 3(k+2)/2  } = \ell^{1 + 3p/4}, \qquad \qquad p > 4,
        \end{split}
    \end{equation}
    which is equal to the upper bound in \eqref{SFUpperBoundS14}. For $2 < p < 4$, we do like in \eqref{TrivialHolderForLowerBound} to get
    \begin{equation}\label{TrivialHolderLowerBound2}
        S_{5/4,p}(\ell)^{1/p} \geq S_{5/4,2}(\ell)^{1/2} \simeq \ell\, \log(\ell^{-1})^{1/2}, \qquad 2 < p < 4,
    \end{equation}
    which matches the upper bound in \eqref{SFUpperBoundS14}. For $p<2$, like in \eqref{InterpolatingBetween2And3_Version2}, using \eqref{SFUpperBoundS14} we get
    \begin{equation}
        S_{5/4,p}(\ell) \geq \frac{S_{5/4,2}(\ell)^{3-p}}{S_{5/4,3}(\ell)^{2-p}} \gtrsim \ell^p \log( \ell^{-1})^{p/2}, \qquad 0 < p < 2,
    \end{equation}
    which is optimal according to \eqref{SFUpperBoundS14}. This concludes the case $s=5/4$.

    \item $5/4 < s < 3/2$: 
    If $2k+2 > p > 2/(3-2s)$, we do like in \eqref{InterpolationGeneralWithHolder} and \eqref{InterpolationGeneral} with $\theta=(2k+4-p)/2$ and we use the result for $p \in 2\mathbb N$ to get
    \begin{equation}
        \begin{split}
        S_{s,p}(\ell) & \geq S_{s,2k+2}(\ell)^\theta \, S_{s,2k+4}(\ell)^{-\theta/\theta'} \simeq_{s,p} \ell^{(1-(k+1)+(2k+2)s)\theta}\, \ell^{(1-(k+2)+(2k+4)s)(1-\theta)} \\ 
        & = \ell^{\theta (1-2s) - k - 1 + (2k+4)s} = \ell^{ps + 1 - p/2}, \qquad p > 2/(3-2s),
        \end{split}
    \end{equation}
    which is the same as the upper bound in \eqref{SFUpperBoundsSBetween14And12}.
    For $2 < p \leq 2/(3-2s)$, using H\"older's inequality like in \eqref{TrivialHolderForLowerBound} gives  
    \begin{equation}\label{SFLowerBoundSBetween14And12ForSmallP}
        S_{s,p}(\ell)^{1/p} \geq S_{s,2}(\ell)^{1/2} \simeq_s \ell, \qquad \qquad 2 < p \leq 2/(3-2s).
    \end{equation}
    This matches the upper bound in \eqref{SFUpperBoundsSBetween14And12} when $p < 2/(3-2s)$.
    For $p<2$, doing like in \eqref{InterpolatingBetween2And3_Version2} we use \eqref{SFwhenPis2} and \eqref{SFUpperBoundsSBetween14And12} to get
    \begin{equation}
        S_{s,p}(\ell) \geq \frac{S_{s,2}(\ell)^{3-p}}{S_{s,3}(\ell)^{2-p}} \gtrsim \ell^p, \qquad \qquad 0 < p < 2,
    \end{equation}
    which according to the upper bound in \eqref{SFUpperBoundsSBetween14And12} is optimal.
    Finally, in the case of the critical $p=2/(3-2s)$, joining \eqref{SFLowerBoundSBetween14And12ForSmallP} and the upper bound in \eqref{SFUpperBoundsSBetween14And12} we get
    \begin{equation}
        \ell \lesssim_p S_{p,s}(\ell)^{1/p} \lesssim \ell\, \log(\ell^{-1})^{1/2}, \qquad \text{ if } p = 2/(3-2s). 
    \end{equation}
    Observe that if $p=2/(3-2s)$ is an even number, in \eqref{even p and critical s} we got a better lower bound $\ell \log(\ell^{-1})^{1/p}$. We expect this logarithmic correction for the lower bound to hold also for non-even $p$.

\end{itemize}
This concludes the proof of Theorem~\ref{TheoremStructureFunctions}.

\section{The multifractal formalism for \texorpdfstring{$R_s$}{Rs}}\label{SECTION_MultifractalFormalism}

Last, we prove Theorem~\ref{Theorem_MultifractalFormalism} and Proposition~\ref{Proposition_MultifractalEta} in Section~\ref{SECTION_Definitions_And_Results}, that is, that the functions $R_s$ satisfy the multifractal formalism. We begin with Proposition~\ref{Proposition_MultifractalEta}, which is a direct consequence of Proposition~\ref{PropDeltak}. Recall that 
\begin{equation}
    \eta(p) = \sup\{ \sigma \mid R_s \in B^{\sigma/p}_{p,\infty}  \}.
\end{equation}
Besov spaces can be described in the Fourier space using the Littlewood-Paley decomposition. As earlier, let $A$ be the Littlewood-Paley annuli parameter. For $k \in \mathbb N$, let us denote by $P_k f$ the Littlewood-Paley projections of $f$ onto $[A^k,A^{k+1})$.
Then,
\begin{equation}
    f \in B^\sigma_{p,\infty} \quad \Longleftrightarrow \quad \sup_{k \in \mathbb N}  A^{k\sigma}\, \lVert P_k f \rVert_{L^p} < \infty.
\end{equation}
\begin{rmk}
The blocks $P_kR_s$ here correspond to the Littlewood-Paley decomposition with parameter $A$. On the other hand, the blocks $\Delta_k$ that we defined in \eqref{LittlewoodPaleyDecompositionInProof} and used in Section~\ref{SECTION_LpNorms} correspond to the decomposition with parameter $A^{1/2}$.
\end{rmk}

\begin{proof}[Proof of Proposition~\ref{Proposition_MultifractalEta}]
Observe first that 
\begin{equation}
    P_k R_s(x) = \sum_{n= A^{k/2}}^{ (A^{k+1}-1)^{1/2}} \frac{e^{2\pi i n^2 x}}{n^{2s}}, 
\end{equation}
so we need to use Proposition~\ref{PropDeltak} with $A^{1/2}$. Then, we have 
\begin{equation}\label{LP_BesovBound}
    A^{k\sigma/p} \, \lVert P_k R_s \rVert_p \simeq  A^{k(\sigma/p + 1/4 - s)}, \qquad \text{ if } p < 4,
\end{equation}
which is bounded in $k$ if and only if $\sigma /p + 1/4 - s \leq 0$, or what is the same, $\sigma \leq p(s - 1/4)$. Thus, $\eta_s(p) = p(s-1/4)$ when $p<4$. The same result is valid for $p=4$ because an extra factor $k^{1/4}$ does not alter the boundedness of \eqref{LP_BesovBound}. On the other hand, also by Proposition~\ref{PropDeltak},  
\begin{equation}
    A^{k\sigma/p} \, \lVert P_k R_s \rVert_p \simeq  A^{k(\sigma/p + 1/2 - 1/p - s)}, \qquad \text{ if } p > 4,
\end{equation}
which remains bounded only when $\sigma \leq 1 + ps - p/2$. Thus, $\eta_s(p) = 1 + p(s - 1/2 )$ when $p>4$.
\end{proof}

Once we know the value of $\eta_s(p)$, we are ready to prove Theorem~\ref{Theorem_MultifractalFormalism}.
\begin{proof}[Proof of Theorem~\ref{Theorem_MultifractalFormalism}]
We need to compute the Legendre transform of $\eta_s$, which we denote by $\eta_s^*$,
and check that it coincides with $d_s$ in \eqref{SpectrumOfSingularitiesOfRs}. Write 
\begin{equation}
    \eta_s^*(\alpha) = \inf_{p>0}\{ \alpha p - \eta(p) + 1 \} = 1 - \sup_{p>0}\{ \eta_s(p) - \alpha p \}.
\end{equation}
Geometrically, we need to compute the supreme distance between the graph of $\eta_s$ and lines $p \mapsto \alpha p$. If $\alpha \geq s-1/4$, then 
$\eta_s(p) - \alpha p  \leq 0$. The equality holds at $p=0$, so $\eta_s^*(\alpha) = 1$. On the other hand, if $\alpha < s-1/2$, then 
$\eta_s(p) - \alpha p \geq 0$ and diverges as $p \to \infty$. Thus, $\eta_s^*(\alpha) = 1 - \infty = -\infty$. If $\alpha = s-1/2$, then $\alpha p$ is still below $\eta_s(p)$, but the two graphs become parallel when $p>4$. The maximum distance is attained at any $p\geq 4$, so $\eta^*(\alpha) = 1-(\eta_s(4)-4\alpha) =0$. Last, if $s-1/2 < \alpha < s-1/4$, then $p \mapsto \alpha p$ crosses the graph of $\eta_s(p)$ and $\eta_s(p) - \alpha p >0$ for small $p$. The maximum is reached at $p=4$, at the point where the slope of $\eta_s(p)$ changes, so $\eta_s^*(\alpha) = 1 - (\eta_s(4) - 4\alpha) = 4\alpha - 4s + 2$. Thus, 
\begin{equation}
    \eta_s^*(\alpha) = \left\{ \begin{array}{ll}
        -\infty, & \alpha < s-1/2, \\
        4\alpha - 4s + 2, & s-1/2 \leq \alpha \leq s-1/4, \\
        1, & \alpha > s-1/4,
        \end{array} \right.
\end{equation}
which matches the spectrum of singularities $d_s(\alpha)$ in \eqref{SpectrumOfSingularitiesOfRs} when $\alpha \leq s-1/4$.
 
\end{proof}

\section*{Acknowledgements}
We would like to acknowledge Alexandre Boritchev and Luis Vega for their support. Many thanks also to Arthur Vavasseur for his aid in some of the technical results in Section~\ref{SECTION_AuxiliaryResults}, to Ioannis Parissis and Luz Roncal for a  discussion on the results of Section~\ref{SECTION_LpNorms} and to Felipe Ponce-Vanegas for discussions on the analytical adaptation of intermittency. 

Daniel Eceizabarrena is supported by the Simons  Foundation  Collaboration  Grant  on  Wave Turbulence (Nahmod’s Award ID 651469). 
Victor Vilaça Da Rocha is supported by NSF grant DMS1800241.

\bibliographystyle{acm}
\bibliography{EV}

\end{document}